\documentclass[10pt]{article}

\usepackage{amsmath,amssymb,amsthm,color,epsfig,mathrsfs}
\usepackage{amsbsy}

\newtheorem{theorem}{Theorem}

\newcounter{geqncount}
    {\refstepcounter{equation}%
     \setcounter{geqncount}{\value{equation}}%
     \setcounter{equation}{0}%
  }%
    {\setcounter{equation}{\value{geqncount}}}

\newcommand{\epsr}{\gamma}
\newcommand{\ek}{\tau}
\newcommand{\ew}{\zeta}

\newcommand{\apbar}{a_\text{\tiny{$\Pbar$}}}
\newcommand{\ap}{a_\text{\tiny{$P$}}}
\newcommand{\omegap}{\omega_\text{\tiny p}}

\newcommand{\xx}{{\mathbf x}}
\newcommand{\yy}{{\mathbf y}}

\newcommand{\kk}{{\hat \kappa}}

\newcommand{\ZZ}{\mathbb{Z}}
\newcommand{\RR}{\mathbb{R}}
\newcommand{\CC}{\mathbb{C}}
\newcommand{\QQ}{{Q}}
\newcommand{\at}[1]{|_{\scriptscriptstyle{#1}}}
\newcommand{\nn}{{\mathbf n}}
\newcommand{\grady}{\nabla\hspace{-2pt}_\yy}
\newcommand{\gradx}{\nabla\hspace{-2pt}_\xx}

\newcommand{\kdotgrad}{\kk\hspace{-1pt}\cdot\hspace{-3pt}\nabla}

\newcommand{\dP}{{\partial P}}
\newcommand{\Pbar}{{P^c}}
\newcommand{\Honeper}{H^1_\text{\scriptsize per}}

\newcommand{\av}[2]{\langle {#1} \rangle_\text{\!\tiny$#2$\,}}

\newcommand{\half}{\frac{1}{2}}
\newcommand{\stsc}{\rightarrow\hspace{-1.7ex}\rightarrow}

\begin{document} 

\bibliographystyle{plain} 

\begin{center}
{\bf \Large  Convergent Power Series for Fields \\ \vspace{1ex} in Positive or Negative High-Contrast Periodic Media}
\end{center}

\vspace{0.2ex}

\begin{center}
{\scshape \large Santiago P. Fortes, Robert P. Lipton, and Stephen P. Shipman} \\
\vspace{1ex}
{\itshape Department of Mathematics, Louisiana State University}
\end{center}

\vspace{3ex}
\centerline{\parbox{0.9\textwidth}{
{\bf Abstract.}\ We obtain convergent power series representations for Bloch waves in periodic high-contrast media.  The material coefficient in the inclusions can be positive or negative.  The small expansion parameter is the ratio of period cell width to wavelength, and the coefficient functions are solutions of the cell problems arising from formal asymptotic expansion.  In the case of positive coefficient, the dispersion relation has an infinite sequence of branches, each represented by a convergent even power series whose leading term is a branch of the dispersion relation for the homogenized medium. 
In the negative case, there is a single branch.}}

\vspace{3ex}
\noindent
\begin{mbox}
{\bf Key words:}  high contrast, double porosity, photonic crystal, negative index, dispersion relation, power series solution, generating function, Bloch wave, homogenization, meta-material.
\end{mbox}
\vspace{3ex}

\hrule
\vspace{1.1ex}

\section{Introduction} 

The objective of this work is the obtention of convergent power-series representations of Bloch waves and dispersion relations for the Helmholtz equation
\begin{equation}\label{Helmholtz0}
  \nabla\cdot a_d(x)\nabla U(\xx) + \frac{\omega^2}{c^2} U(\xx) \,=\, 0\,,
\end{equation}
in which the scalar material coefficient $a_d(\xx)$ has small period $d$ in each spatial direction and consists of two highly contrasting phases, with either positive or negative contrast
\begin{equation*}
  a_d(\xx) =
    \renewcommand{\arraystretch}{1.5}
\left\{
  \begin{array}{ll}
    \pm d^2/\epsr & \text{for } \xx\in d(P+m),\, m\in\ZZ^n, \\
    1 & \text{for } \xx\in d(\Pbar + m),\, m\in\ZZ^n.
  \end{array}
\right.
\end{equation*}
The ``inclusion" $P$ is a subset of the unit cube $\QQ$ in $\RR^n$ ($n\geq2$) with $C^1$ boundary $\dP$, such that the complement $\Pbar=\QQ\setminus P$ contains the boundary sides of $\QQ$ (Fig.~\ref{fig:inclusion}).
Both $\xx$ and $d$ carry units of length so that $\xx/d\in\RR^n$.  The fixed constant $\epsr$ carries units of area, $\omega$ is the operating frequency, and $c$ is a fixed reference celerity.
We will develop power series in the expansion parameter $d$ or the parameter of quasi-staticity, which measures the ratio of cell size to wavelength.

This equation has applications to acoustic waves in porous media, high-contrast photonic crystals, and plasmonic crystals, in the regime of large wavelength-to-cell ratio.  In the former, the inclusion contains the soft phase and the host material is of a stiff phase.  In the latter two applications, the equation is a two-dimensional reduction of the Maxwell system in a periodic array of high-contrast rods, the field $u$ is the out-of-plane component of the magnetic field, and $a_d(\xx) = \epsilon^{-1}(\xx)$ is the reciprocal of the dielectric coefficient.  The rods consist of a lossless dielectric material in the case of a positive coefficient or a lossless plasma in the case of a negative coefficient.  

As $d$ tends toward zero in the positive case, a multi-branched homogenized, or quasi-static, dispersion relation relating frequency to Bloch wavevector emerges.  This relation, which was obtained by Zhikov \cite[\S8]{Zhikov2000},\cite{Zhikov2004}, involves the subset $\{\mu_n\}$ of the set of all Dirichlet eigenvalues of the Laplacian $-\Delta=-\nabla\!\cdot\!\nabla$ in $P$ for which the means $\av{\phi_n}{P}$ of the corresponding eigenfunctions $\phi_n$ do not vanish,
\begin{equation}\label{dispersion0}
  k^2 = \text{const}\cdot\frac{\omega^2}{c^2}\sum_{n=1}^\infty \frac{\mu_n\av{\phi_n}{P}^2}{\mu_n-\epsr\omega^2/c^2}.
\end{equation}
Here, $k$ is the wavenumber and the constant depends on the direction of the wavevector.
The relation is graphed in Fig.~\ref{fig:DispersionPos}.  It reveals a sequence of spectral bands $[\mu^*_n,\mu_n]$ characterized by those points on the abscissa for which the right-hand side of \eqref{dispersion0} is nonnegative.

In the two-dimensional application to electromagnetic fields, the $\mu_n$ correspond to resonances of an effective magnetic permeability of the homogenized medium, as shown by Bouchitt\'e and Felbacq \cite{BouchitteFelbacq2004,FelbacqBouchitte2005}.  Modes that are localized in defects within a high-contrast medium of this kind have been shown to converge spectrally by Kamotski and Smyshlyaev~\cite{KamotskiSmyshlyaev2006} as well as Cherdantsev~\cite{Cherdantse2009} using homogenization techniques.
Limits of spectra in high-contrast media have also been obtained when the value of $a$ in a periodic thin grating is small \cite{FigotinKuchment1996,FigotinKuchment1998} or when $a$ is small in the periodically dispersed regions enclosed by the grating
\cite{Zhikov2004}.

The remaining Dirichlet eigenvalues $\{\mu'_\ell\}$ of $-\Delta$ in $P$, that is, those whose eigenfunctions all have mean zero, also play a role in the small-$d$ limit of equation \eqref{Helmholtz0}.
Hempel and Lienau \cite{HempelLienau2000} proved that the spectra of the ``double porosity" operators $A_d = -\nabla\cdot a_d(\xx)\nabla U(\xx)$ converge, as $d\to0$, to the union of the bands $[\mu^*_n,\mu_n]$ plus those eigenvalues $\mu'_\ell$ that fall within the gaps between these bands.  (In \cite{HempelLienau2000}, the authors actually deal with a fixed period and a material coefficient $a$ that is fixed in the inclusions $P\!+\!m$ and tends to infinity in the host material; by a scaling of $a$ and $\xx$, the spectrum is seen to be equal to that of $A_d$.)  Zhikov deals directly with the double porosity operator and identifies the limit of the spectra as $d\to0$ in the sense of Hausdorff with the spectrum of a limiting operator $A$ in the sense of two-scale convergence \cite[\S3]{Zhikov2004}, where the two scales are $\xx$ and $\yy=\xx/d$.

At values of $\ew=\gamma\omega^2/c^2$ belonging to a branch of relation \eqref{dispersion0}, the equation $(A-\ew)u=0$ admits plane-wave solutions in $x$ with a computable periodic microscopic variation in $y\in\QQ$ which is constant and nonzero in the host $\Pbar$ and variable in the inclusion $P$.
On the other hand, at any eigenvalue $\ew=\mu_\ell$ or $\ew=\mu'_\ell$, $(A-\ew)u=0$ has solutions that are supported in the inclusions alone and vanish in the host.  If $\mu'_\ell$ falls in a gap of the homogenized dispersion relation, one views it as a band that has degenerated to a point, and if $\mu'_\ell$ falls in a band of the homogenized dispersion relation, one views it as a gap that has degenerated to a point (see \cite{HempelLienau2000} for a discussion of this).  Generically, eigenfunctions have nonzero mean, but for the important example of a circle (as illustrated in Fig.~\ref{fig:DispersionPos}) or a sphere, all but the radially symmetric eigenfunctions have zero mean.

\begin{figure}[t]
\centerline{\scalebox{0.3}{\includegraphics{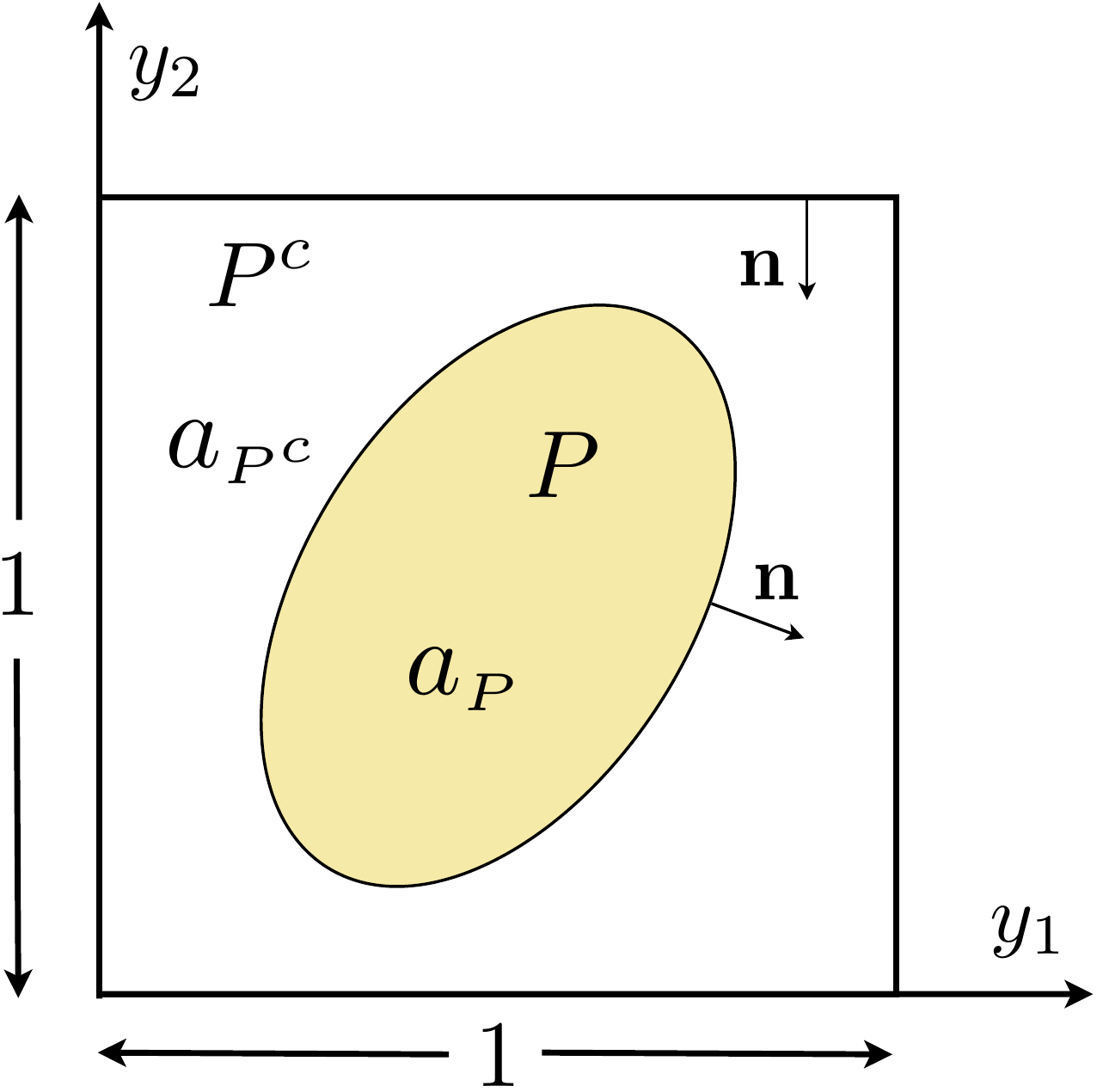}}}
\caption{\small The unit cell in $\RR^2$ with a high-contrast inclusion.}
\label{fig:inclusion}
\end{figure}
%

\begin{figure}[h]
\centerline{\scalebox{0.7}{\includegraphics{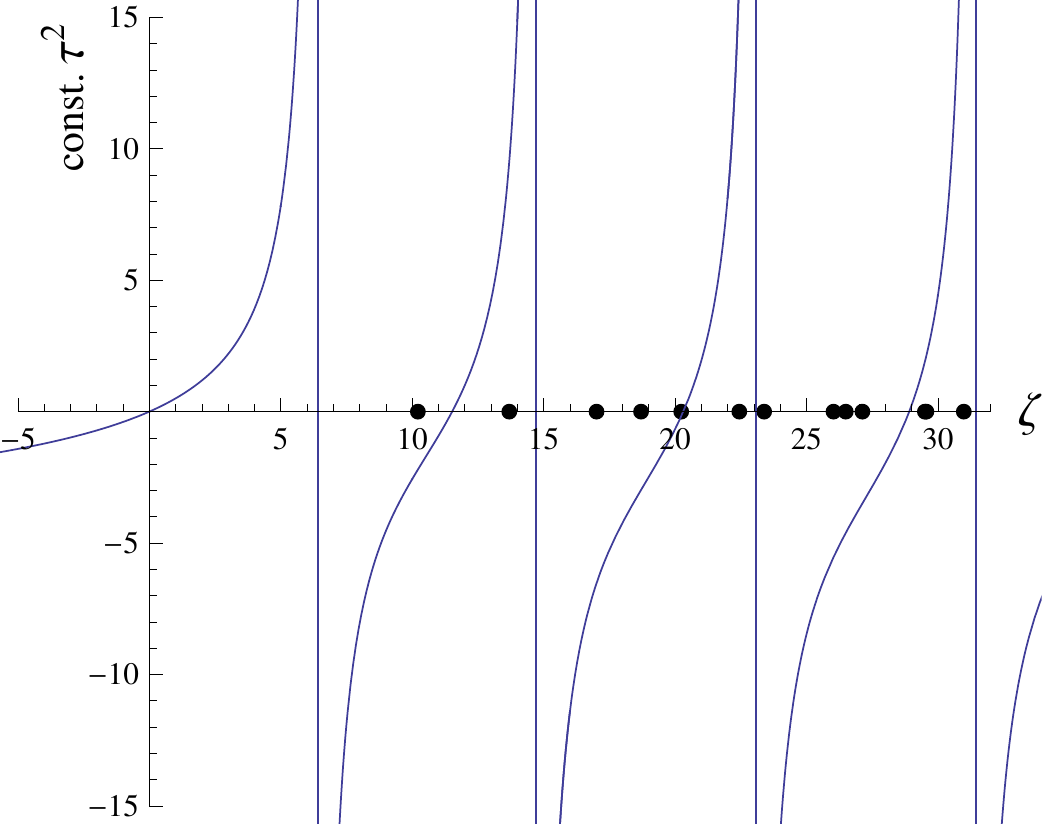}}}
\caption{\small The dispersion relation of the homogenized medium for an array of circles of radius 0.375 with positive material coefficient ($\ek^2=\gamma k^2$ and $\ew=\gamma\omega^2/c^2$ are the reduced square wavenumber and frequency).  The asymptotes occur at the eigenvalues ${\mu_n}$ of $-\Delta$ that correspond to eigenfunctions $J_0(\mu_n r)$, the spectral bands are the intervals $[\mu^*_n,\mu_n]$ where the function is positive, and the dots indicate the other eigenvalues $\mu'_\ell$.}
\label{fig:DispersionPos}
\end{figure}

The guiding principle in the positive case is this:

\smallskip
{\em
For each wavevector $k\kk$ ($|\kk|=1$), there exists a unique number $\ew=\gamma\omega^2/c^2\in[\mu^*_n,\mu_n]$ in the $n^\text{th}$ spectral band defined by \eqref{dispersion0}.  If $\ew$ is not a Dirichlet eigenvalue for $-\Delta$ in $P$, then, for $\eta=kd$ sufficiently small, the structure supports a Bloch wave with wavevector $k\kk$ and square frequency $\omega^2=c^2\ew^\eta/\epsr$.  As $\eta\to0$, $\ew^\eta$ converges to $\ew$ and the Bloch wave converges weakly to a plane wave with wavevector $k\kk$ and square frequency $\omega^2=c^2\ew/\epsr$ and has a computable strong two-scale limit.  {\bfseries The Bloch wave and the frequency $\ew^\eta$ at nonzero values of $\eta$ are given by convergent power series in $\eta$.}
On the other hand, for each Dirichlet eigenvalue $\mu'_{\ell}$ for $-\Delta$ in $P$ for which all eigenfunctions have mean zero,
and for each $k\kk$, the structure supports a Bloch wave whose (scaled) square frequency $\ew^\eta$ is given to leading order by $\mu'_{\ell}$ as $\eta\rightarrow 0$.
}

\smallskip

The demonstration of the convergent power series representation of the Bloch fields and the frequency, for both positive and negative coefficient in $P$, is the subject of this work.  The case of negative material coefficient in $P$ is simpler, as the homogenized dispersion relation has only one branch (Fig.~\ref{fig:DispersionNeg}) and the eigenvalues of $P$ do not play a role.  We present all details for the case of positive coefficient for which the leading term $\ew_0$ of the power series for $\ew^\eta$ lies in a band of the homogenized dispersion relation and does not coincide with $\mu'_\ell$.  In the case that $\ew_0=\mu'_\ell$ for some $\ell$, we allow the remark at the end of section \ref{sec:sequence} to suffice.

Our approach to this problem is to expand the frequency and field in $\eta$ and then to prove convergence of the resulting power series to a dispersion relation for the structure and the associated Bloch fields {\em for nonzero values of $\eta$} within some radius of convergence.  The analysis extends our recent work on high-contrast plasmonic inclusions \cite{Fortes2010,FortesLiptonShipman2010}, in which we establish convergent power series representations for fields and dispersion relations in sub-wavelength plasmonic crystals (dispersions of rods consisting of a plasma material).  An advantage of the method is that {\em it does not rely on coercivity of the two-scale operator}, and one is therefore able to prove existence of solutions for structures with {\em large negative contrast} in the material coefficient.  The advantage of a power series solution is that it provides an approximation of the true solution up to any algebraic order in $\eta$.
For certain coercive problems, it has been shown that the formal power-series expansion is an asymptotic series for the solution \cite{SmyshlyaevCherednich2000,Panasenko2009} and even that suitable truncation yields exponentially small error \cite{KamotskiMatthiesSmyshlyaev2007}.

Convergent series for fields in composites with highly contrasting complex conductivity have been obtained by Bruno \cite{Bruno1991}.  As far back as 1890, Hermann Schwarz \cite{Schwarz1890} obtained a power series solution of the equation $\Delta u + \lambda pu=0$ in a bounded domain, where $p(x)>0$.  The expansion parameter is~$\lambda$, and the solution has a radius of convergence equal to the first eigenvalue.  In fact, he also represents the first eigenfunction at the radius of convergence by a convergent sequence.

\smallskip
Our strategy is as follows.
\begin{enumerate}
  \item Fix a wavevector $k\kk$ and expand the field and frequency in power series in $\eta=kd$ to obtain an infinite coupled sequence of PDEs for the coefficients.
  \item Solve the first few equations to obtain the homogenized (quasi-static) dispersion relation and fields.
  \item Choose a branch of the homogenized dispersion relation, which corresponds to a unique value of $\ew$ and iteratively solve the infinite sequence of PDEs for the coefficients of the power series.  In the positive case, it can observed explicitly how solubility of the system is obstructed by secular terms in the expansion if $\ew$ is a Dirichlet eigenvalue of $-\Delta$ in the inclusion $P$.
  \item Obtain a nonzero radius of convergence of the series by establishing an exponential bound on the coefficients.  This is achieved through the use of systems of generating functions.  The radius is uniform over $k$ in the negative case, and it is uniform in the positive case if $\ew$ is bounded away from the eigenvalues $\mu'_\ell$.
  \item Prove that the power series solve the field equation and that the odd coefficients of the series for the dispersion relation vanish.
\end{enumerate}

{\slshape {\bfseries A remark} on proof of convergence and comparison with plasmonic crystals.} \
In \cite{FortesLiptonShipman2010}, we obtained power-series solutions for the dispersion relation and the associated Bloch waves for two-dimensional plasmonic crystals, in which the plasma frequency $\omegap$ tends to infinity as the inverse of the cell width.  Following the scaling $\omegap=c/d$ of Shvets and Urzhumov \cite{ShvetsUrzhumov2005}, the small $\eta$ regime leads to a large negative dielectric coefficient in the cylindrical plasmonic inclusions $\epsilon_P = 1-1/(\xi^2\eta^2)$, where $\xi = \omega^2/(c^2k^2)$.
In that work, we use fine properties of the Catalan numbers to obtain quantitative values for a lower bound on the radius of convergence \cite[\S6bc]{FortesLiptonShipman2010} by obtaining an explicit exponential bound on the norms of the coefficients of the expansions.   Special properties of the two-scale limit of the Bloch fields that were available in the plasmonic case but not in the case of positive contrast also played a crucial role in the analysis.
The method of generating functions that we use in the present work is simpler and applicable as well to the case of plasmonic crystals, although it does not provide a way to estimate the radius of convergence.  This approach can in fact be applied quite generally when one is able to obtain recursive bounds on the solutions of the higher-order cell problems, as, for example, for arrays of micro-resonators \cite{Shipman2010a}.
In \cite{FortesLiptonShipman2010}, we assumed symmetry of the inclusion under a rotation of $180^\circ$.  We do not make any symmetry assumption in the present paper, and our analysis here eliminates the need for this assumption for plasmonic crystals as well.

\section{Formal power series and dispersion relations} 

We begin by fixing a wavenumber and introducing power series expansions of a Bloch field and its frequency in the parameter $\eta=dk$.  Solving for the first two coefficient fields and then imposing the solvability condition at the next order in the matrix material leads to the homogenized dispersion relation.

\subsection{Long-wavelength Bloch waves}

The Helmholtz equation \eqref{Helmholtz0} can be written as the system
\begin{equation}\label{systemx}
  \renewcommand{\arraystretch}{1.2}
\left.
  \begin{array}{cl}
    \apbar\Delta_\xx U + \frac{\omega^2}{c^2}U = 0 & \text{in the matrix,} \\
    \ap\Delta_\xx U + \frac{\omega^2}{c^2} U = 0 & \text{in the inclusion,} \\
    \apbar\gradx U|_{\Pbar} \cdot\nn = \ap\gradx U|_P \cdot\nn & \text{on the interface.}
  \end{array}
\right.
\end{equation}
We investigate Bloch-wave solutions---solutions that are periodic at the microscopic scale and modulated by a plane wave at the macroscopic scale,
\begin{equation*}
  U(\xx) = u(\yy) e^{ik\kk\cdot\xx}, \quad \xx = d\,\yy,
\end{equation*}
in which $u(\yy)$ has the unit cube $\QQ=[0,1]^n$ as a period cell.
Apart from the definitions of $\ap$ and $\apbar$, the equations for $u(\yy)$ explicitly depend only on the unit vector $\kk$, the ratio $\omega^2/(k^2c^2)$, and the parameter $\eta$,
\begin{equation}\label{systemy}
  \renewcommand{\arraystretch}{1.2}
\left\{
  \begin{array}{ll}
    \apbar\,k^2(\Delta_\yy + \eta\, 2i\kk\cdot\grady - \eta^2)u + \eta^2\frac{\omega^2}{c^2} u = 0, & \yy\in \Pbar, \\
    \ap\,k^2(\Delta_\yy + \eta\, 2i\kk\cdot\grady - \eta^2)u + \eta^2\frac{\omega^2}{c^2} u = 0, & \yy\in P, \\
    \apbar(\grady+i\eta\kk)u|_{\Pbar}\cdot\nn = \ap(\grady+i\eta\kk)u|_P\cdot\nn, & \yy\in \partial P.
  \end{array}
\right.
\end{equation}
At this point, one can consider different asymptotic regimes depending on how the coefficients are defined.  In this work, we take $\apbar$ to be unity and $\ap$ to be proportional to the area of the unit cell.
\begin{equation}
  \apbar = 1, \quad \ap = \frac{d^2}{\epsr} = \frac{\eta^2}{\epsr k^2},
\end{equation}
where $\epsr$ is a constant carrying the dimension of area.  It is also convenient to introduce the non-dimensional wavenumber $\ek$ and square frequency $\ew$ by
\begin{equation*}
  \ek^2 = \epsr k^2, \qquad \ew = \epsr\frac{\omega^2}{c^2},
\end{equation*}
which transforms the system \eqref{systemy} into
\begin{equation}\label{master}
  \renewcommand{\arraystretch}{1.2}
\left\{
  \begin{array}{ll}
    \ek^2 (\Delta + \eta\, 2i\kdotgrad - \eta^2)u + \eta^2\ew u = 0, & \yy\in \Pbar, \\
    (\Delta + \eta\, 2i\kdotgrad - \eta^2)u + \ew u = 0, & \yy\in P, \\
    \ek^2 (\nabla+i\eta\kk)u\at{\Pbar}\cdot\nn = \eta^2(\nabla+i\eta\kk)u\at{P}\cdot\nn, & \yy\in \partial P.
  \end{array}
\right.  
\end{equation}

\subsection{Power series expansions}

\noindent
By inserting the power series ansatz
\begin{eqnarray}
  && u^\eta = u_0 + \eta u_1 + \eta^2 u_2 + \cdots \label{ueta}\\
  && \ew^\eta = \ew_0 + \eta\ew_1 + \eta^2\ew_2 + \cdots\label{xieta}
\end{eqnarray}
into the system \eqref{master},
one obtains a system of coupled partial differential equations for the coefficients, 
\begin{equation}\label{hmstrong}
  \renewcommand{\arraystretch}{1.5}
\left\{
  \begin{array}{ll}
    \ek^2(\Delta u_m + 2i\kdotgrad u_{m-1} - u_{m-2}) + \sum_{\ell=0}^{m-2} \ew_\ell u_{m-2-\ell}=0 & \text{in $\Pbar$,} \\
    \Delta u_m + 2i\kdotgrad u_{m-1} - u_{m-2} + \sum_{\ell=0}^{m} \ew_\ell u_{m-\ell}=0 & \text{in $P$,} \\
    \ek^2 (\nabla u_m + i\kk u_{m-1})\at{\Pbar}\cdot\nn
       = (\nabla u_{m-2} + i\kk u_{m-3})\at{P}\cdot\nn & \text{on $\dP$,}
  \end{array}
\right.
\end{equation}
in which $u_m\equiv0$ for $m<0$.
%
%
For the function $u_0$ in the matrix $\Pbar$, \eqref{hmstrong} yields the BVP
\begin{equation*}
\renewcommand{\arraystretch}{1.2}
\left.
  \begin{array}{ll}
    \Delta u_0 = 0 & \text{ in } \Pbar, \\
    \nabla u_0\cdot\nn =0 & \text{ on } \partial P,      
  \end{array}
\right.
\end{equation*}
from which we infer that $u_0$ is a constant, which we denote by $\bar u_0$, in $\Pbar$.
The system \eqref{hmstrong} yields the following boundary-value problem for $u_0$ in $P$:
\begin{equation}\label{h0strong}
  \renewcommand{\arraystretch}{1.4}
\left.
  \begin{array}{rl}
    \Delta u_0 + \ew_0 u_0 = 0 & \text{ in } P, \\
    u_0\at{P} = \overline{u}_0 & \text{ on } \partial P.
  \end{array}
\right.
\end{equation}
At this stage \eqref{h0strong} presents us with an alternative: either $\ew_0$ coincides with a Dirichlet eigenvalue, in which we take $\bar u_0=0$ if $\ew_0=\mu'_\ell$ for some $\ell$; otherwise $\bar u_0\not=0$.
In the former case, whether $\ew_0$ falls within a spectral band or a spectral gap, one obtains a power series for a Bloch wave for each wavevector (see the remark in section \ref{sec:sequence}) for which the leading-order term in the expansion of $\ew^\eta$ is given by $\mu'_\ell$.
In the latter case, the possible frequencies $\ew_0$ are determined by $k\kk$.
At this point, we make the assumption that $\ew_0$ is not an eigenvalue and that $\bar u_0\not=0$, which means that we are seeking a power series for a field that nonzero in the matrix.  


It is convenient to work with the dimensionless fields $\psi_m$ defined through
\begin{equation*}
  u_m = i^m\bar u_0\psi_m,
\end{equation*}
so that $\psi_0=1$ in $\Pbar$ and the system \eqref{hmstrong} becomes
\begin{equation}\label{psimstrong}
  \renewcommand{\arraystretch}{1.5}
\left\{
  \begin{array}{ll}
    \ek^2(\Delta\psi_m + 2\kdotgrad\psi_{m-1} + \psi_{m-2}) - \sum_{\ell=0}^{m-2} (-i)^\ell\ew_\ell\,\psi_{m-2-\ell}=0
                                   & \text{in $\Pbar$,} \\
    (\Delta+\ew_0)\psi_m + 2\kdotgrad\psi_{m-1} + \psi_{m-2} + \sum_{\ell=1}^{m} (-i)^\ell\ew_\ell\,\psi_{m-\ell}
    =0 
                                   & \text{in $P$,} \\
    \ek^2 (\nabla\psi_m + \kk\psi_{m-1})\at{\Pbar}\cdot\nn
       = -(\nabla\psi_{m-2} + \kk\psi_{m-3})\at{P}\cdot\nn & \text{on $\dP$}.
  \end{array}
\right.
\end{equation}
Notice that the equation for $\psi_m$ in $P$ is a Dirichlet boundary-value problem for the operator $\Delta+\ew_0$ with boundary data equal to the trace of $\psi_m$ in $\Pbar$ and interior forcing coming from
$\psi_\ell$ in $P$ for $\ell\!<\!m$ and $\ew_\ell$ for $\ell\leq m$.  The equation for $\psi_m$ in $\Pbar$, on the other hand, is a Neumann boundary-value problem for the Laplacian with periodic conditions on the boundary of the cube $\QQ$ and normal derivative specified by those of $\psi_{m-1}$ in $\Pbar$ and $\psi_\ell$ in $P$ for $\ell\leq m\!-\!2$ and interior forcing coming from $\psi_\ell$ in $\Pbar$ for $\ell\!<\!m$ and $\ew_\ell$ for $\ell\leq m\!-\!2$.

The system \eqref{psimstrong} yields the following boundary-value problem for $\psi_0$ in $P$:
\begin{equation}\label{psi0strong}
  \renewcommand{\arraystretch}{1.4}
\left.
  \begin{array}{rl}
    \Delta \psi_0 + \ew_0\psi_0 = 0 & \text{ in } P, \\
    \psi_0\at{P} = 1 & \text{ on } \partial P.
  \end{array}
\right.
\end{equation}
This problem has a unique solution if $\ew_0$ is not a Dirichlet eigenvalue of $-\Delta$ in~$P$.  The equation for $\psi_1$ in $\Pbar$~is
\begin{equation}\label{psi1strong}
  \renewcommand{\arraystretch}{1.4}
\left.
  \begin{array}{rl}
    \Delta\psi_1 = 0 & \text{ in } \Pbar, \\
    \nabla\psi_1\!\cdot\!\nn \,+\, \kk\!\cdot\!\nn = 0 & \text{ on } \partial P,
  \end{array}
\right.
\end{equation}
which has a unique solution subject to the mean-zero condition
\begin{equation*}
  \int_{\Pbar} \psi_1 = 0.
\end{equation*}
The equation for $\psi_2$ in $\Pbar$ is
\begin{equation*}
  \renewcommand{\arraystretch}{1.3}
\left.
  \begin{array}{rl}
    \ek^2(\Delta\psi_2 + 2\kdotgrad\psi_1 + 1) - \ew_0 = 0 & \text{ in }\Pbar, \\
    \ek^2(\nabla\psi_2 + \kk\psi_1)\at{\Pbar}\!\cdot\!\nn = -\nabla\psi_0\at{P}\!\cdot\!\nn
       & \text{ on }\partial P,
  \end{array}
\right.
\end{equation*}
which is a Neumann problem whose solvability is subject to the condition
\begin{equation}\label{solvability2}
  \ek^2\int_\Pbar (\kdotgrad\psi_1 + 1) - \ew_0\!\int_\QQ \psi_0 = 0\,.  \qquad \text{(solvability)}
\end{equation}
This is implicitly the leading order of the dispersion relation, giving $\ek^2$ as a function of $\ew_0$ (equivalently, $k$ as a function of $\omega_0$).  Indeed, the coefficient multiplying $\ek^2$ is nonzero:  By \eqref{psi1strong}, we obtain
\begin{equation*}
  \int_\Pbar|\nabla\psi_1|^2 = -\int_\Pbar\kk\cdot\nabla\psi_1
  < |\Pbar|^\half \Big( \int_\Pbar |\nabla\psi_1|^2 \Big)^\half,
\end{equation*}
in which the inequality is strict because, since $\psi_1(\yy)$ is periodic and $\kk\cdot\yy$ is not, $\nabla\psi_1$ is not a multiple of $\kk$.
Thus, $\int_\Pbar|\nabla\psi_1|^2 < |\Pbar|$ and we obtain
\begin{equation*}
  \int_\Pbar(\kk\!\cdot\!\nabla\psi_1+1) = \int_\Pbar \kk\!\cdot\!\nabla\psi_1 \,+\, |\Pbar|
  = -\int_\Pbar |\nabla\psi_1|^2 + |\Pbar| > 0.
\end{equation*}

Following Zhikov \cite{Zhikov2000}, we can make the dispersion relation explicit by writing the solution $\psi_0$ in $P$ in terms of the Dirichlet spectral data,
\begin{equation}\label{psi0}
  \psi_0(\yy) \,=\, 1 + \ew_0\sum_{n=1}^\infty \frac{\av{\phi_n}{P}}{\mu_n-\ew_0}\,\phi_n(\yy)
    \,=\, \sum_{n=1}^\infty \frac{\mu_n\av{\phi_n}{P}}{\mu_n-\ew_0}\,\phi_n(\yy),
  \quad \yy\in P,
\end{equation}
in which $\mu_n$ are those Dirichlet eigenvalues of $-\Delta$ in $P$ whose $L^2$-normalized eigenfunctions $\phi_n$ appear in the expansion of the constant function in $P$ and thus have nonzero mean,
\begin{equation*}
  \av{\phi_n}{P} = \int_P \phi_n(\yy)d\yy \not=0.
\end{equation*}
The solvability condition \eqref{solvability2} becomes
\begin{equation}\label{dispersion}
  \ek^2\!\int_\Pbar(\kdotgrad\psi_1+1) = \ew_0\sum_{n=1}^\infty \frac{\mu_n\av{\phi_n}{P}^2}{\mu_n-\ew_0}.
\end{equation}
As observed by Zhikov \cite{Zhikov2000,Zhikov2004}, this relation, shown in Fig.~\ref{fig:DispersionPos} defines an infinite sequence $\{\ew_0^{(n)}, n=0,1,2,\dots\}$ of values of $\ew_0$ such that
\begin{equation*}
  \mu_n < \mu^*_n \leq \ew_0^{(n)} < \mu_{n+1}.
\end{equation*}
The functions $\ew_0^{(n)}=\ew_0^{(n)}(\kk,\ek)$ of $\kk$ and $\ek$ define the leading order in $\eta$ of an infinite sequence of branches of the dispersion relation and reveals a sequence of stop and pass bands for a homogenized medium.   A convergent series representation for $\ew^{(n)}$ for $0\leq\eta<R_n(\ek)$,
\begin{equation*}
  \ew^{(n)} = \ew_0^{(n)} + \eta \ew_1^{(n)} + \eta^2 \ew_2^{(n)} + \dots,
\end{equation*}
will be proved in section~\ref{subsec:convergence} subject to the condition that $\ew_0^{(n)}\not=\mu'_\ell$ for all $\ell$.

\subsection{Negative material coefficient}

If we take the material coefficient in the inclusion to be negative,
\begin{equation*}
  \ap = -\frac{d^2}{\epsr} = -\frac{\eta^2}{\epsr k^2},
\end{equation*}
the plus sign before the $\ew_0$ in the system \eqref{psi0strong} becomes a minus sign, whereas the solvability condition \eqref{solvability2} remains unaltered.  The resulting dispersion relation is
\begin{equation}
  \ek^2\!\int_\Pbar(\kdotgrad\psi_1+1) = \ew_0\sum_{n=1}^\infty \frac{\mu_n\av{\phi_n}{P}^2}{\mu_n+\ew_0},
\end{equation}
which differs from that for positive $\ap$ only by the plus sign in the denominator.  The graph of this relation is obtained from Fig.~\ref{fig:DispersionPos} by reflection about the origin, and the result is that there is only a single dispersion relation that passes through the origin (Fig.~\ref{fig:DispersionNeg}), as we require that $\ew=\gamma\omega^2/c^2$ be positive.

\begin{figure}[h]
\centerline{\scalebox{0.7}{\includegraphics{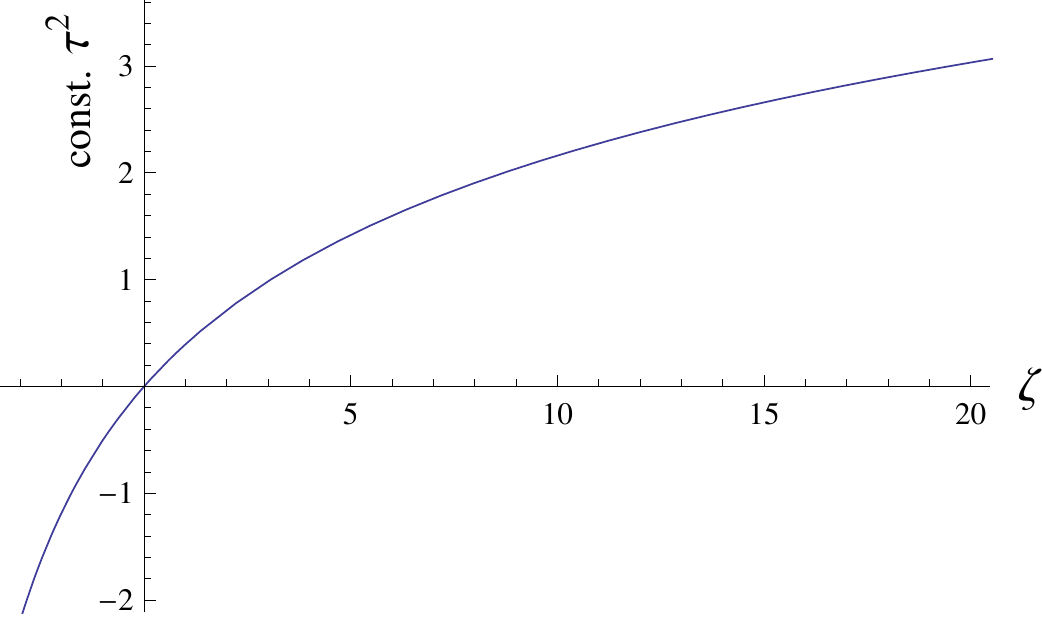}}}
\caption{\small The dispersion relation for the homogenized medium for an array of circles of radius 0.375 with negative material coefficient.  All positive values of $\ew$ are in the spectrum.}
\label{fig:DispersionNeg}
\end{figure}

\section{Solution of the sequence of cell problems}\label{sec:sequence}

We shall develop the mathematics and state the theorems in the case of positive material coefficient in the inclusion, which is, in fact, the more complicated situation because of the resonances $\{\mu'_n\}$.  In section \ref{sec:negative}, we discuss the modifications that need to be made in the case of negative material coefficient.

When solving iteratively for the coefficients of the field and frequency, we do not assume {\it a priori} that the coefficients of the frequency are real.  After proving that the series converges and solves the Helmholtz equation, we then prove {\it a posteriori} that the coefficients are real and that the odd ones vanish.

The weak form of the cell problem for $u$ is
\begin{multline}\label{masterweak}
  \ek^2\!\int_\Pbar \left( \nabla u\cdot\nabla\bar v + i\eta\kk\cdot(u\nabla\bar v-\bar v\nabla u) + \eta^2u\bar v \right)
  -\eta^2\ew\!\int_\Pbar u\bar v \;+\\
  + \eta^2\!\int_P \left( \nabla u\cdot\nabla\bar v + i\eta\kk\cdot(u\nabla\bar v-\bar v\nabla u) + \eta^2u\bar v \right)
  -\eta^2\ew\!\int_P u\bar v \,=\,0
  \qquad \forall\,v\in\Honeper(\QQ),
\end{multline}
and the weak form of the sequence of cell problems \eqref{psimstrong} for $\psi_m$ is
\begin{multline}\label{Ebar}
  \ek^2\!\int_{\Pbar} \nabla\psi_m\cdot\nabla\bar v
  \,+\, \ek^2\!\int_{\Pbar}\left[ \kk\psi_{m-1}\cdot\nabla\bar v -\left(\kk\cdot\nabla\psi_{m-1}+\psi_{m-2}\right)\bar v \right] + \int_\QQ \sum_{\ell=0}^{m-2} (-i)^\ell\ew_\ell\,\psi_{m-2-\ell}\,\bar v \,+ \\
  -\int_P\nabla\psi_{m-2}\cdot\nabla\bar v
  -\int_P \left[ \kk\psi_{m-3}\cdot\nabla\bar v - \left(\kk\cdot\nabla\psi_{m-3}+\psi_{m-4}\right)\bar v\right] \,=\, 0 \qquad
  \forall\; v\in \Honeper(\QQ),
\end{multline}
in which $\psi_m$ and $\ew_m$ are set to zero for $m<0$.

\smallskip
\noindent
{\bfseries\slshape Solving for $\psi_m$ in $P$.}\;
By restricting to test functions $v$ with support in $P$ and putting $m\mapsto m+2$ in \eqref{Ebar}, one obtains a Dirichlet boundary-value problem for $\psi_m$ in $P$, in terms of the boundary values of $\psi_m$ in $\Pbar$:
\begin{equation}\label{E}
  \renewcommand{\arraystretch}{1.5}
\left\{
  \begin{array}{l}
    \displaystyle \int_P\left( -\nabla\psi_m\cdot\nabla\bar v + \ew_0\,\psi_m\bar v \right)
    -\int_P\left[ \kk\psi_{m-1}\cdot\nabla\bar v -\left( \kk\cdot\nabla\psi_{m-1}+\psi_{m-2} \right)\bar v \right] \,+\hspace{5em}\\
    \displaystyle \hfill +\int_P \sum_{\ell=1}^{m} (-i)^\ell \ew_\ell\,\psi_{m-\ell}
    \,\bar v \,=\, 0
    \qquad \forall\; v\in H^1_0(P),\\
    \psi_m|_{\dP^-} = \psi_m|_{\dP^+}.
  \end{array}
\right.
\end{equation}
For $m\geq1$, it is convenient to decompose $\psi_m$ in $P$ as follows (we define $\tilde\psi_0=0$):
\begin{equation}
  \psi_m = \tilde\psi_m + (-i)^m\ew_m\psi_*,
\end{equation}
in which $\tilde\psi_m$ solves a Dirichlet problem involving $\psi_m$ in $\Pbar$, $\psi_\ell$ in $P$ for $\ell<m$, and $\ew_\ell$ for $\ell<m$,
\begin{equation}\label{Etilde}
  \renewcommand{\arraystretch}{1.5}
\left\{
  \begin{array}{l}
    \displaystyle \int_P\left( -\nabla\tilde\psi_m\cdot\nabla\bar v + \ew_0\,\tilde\psi_m\bar v \right)
    -\int_P\left[ \kk\psi_{m-1}\cdot\nabla\bar v -\left( \kk\cdot\nabla\psi_{m-1}+\psi_{m-2} \right)\bar v \right] \,+\hspace{5em}\\
    \displaystyle \hfill +\int_P \sum_{\ell=1}^{m-1} (-i)^\ell \ew_\ell\,\psi_{m-\ell}
    \,\bar v \,=\, 0
    \qquad \forall\; v\in H^1_0(P),\\
    \tilde\psi_m|_{\dP^-} = \psi_m|_{\dP^+}.
  \end{array}
\right.
\end{equation}
and $\psi_*$ is independent of $m$ and satisfies
\begin{equation}\label{Estar}
  \renewcommand{\arraystretch}{1.5}
\left\{
  \begin{array}{l}
  \displaystyle \int_P\left( -\nabla\psi_*\cdot\nabla\bar v + \ew_0\,\psi_*\bar v \right)
  + \int_P\psi_0\bar v \,=\, 0 \qquad \forall\; v\in H^1_0(P),\\
   \psi_*|_{\dP} = 0.
  \end{array}
\right.
\end{equation}
The strong form of this equation is
\begin{equation*}
\renewcommand{\arraystretch}{1.3}
\left.
  \begin{array}{ll}
      (\Delta + \ew_0)\psi_* = -\psi_0 & \text{ in } \; P, \\
      \psi_* = 0 & \text{ on } \; \dP,
  \end{array}
\right.
\end{equation*}
and, given the expression \eqref{psi0} for $\psi_0$, one can always solve explicitly for $\psi_*$,
\begin{equation}\label{psistar}
  \psi_* = \sum_{n=1}^\infty \frac{\mu_n\,\av{\phi_n}{P}}{(\mu_n-\ew_0)^2}\,\phi_n\,
  \qquad \text{in $P$.}
\end{equation}
On the other hand, equation \eqref{Etilde} for $\tilde\psi_m$ can typically only be solved if $\ew_0\not=\mu'_\ell$ for all $\ell$, and this is precisely where this assumption plays its role.

\smallskip
\noindent
{\bfseries\slshape Solving for $\psi_m$ in $\Pbar$.}\;
The field $\psi_m$ in $\Pbar$ is to be determined by \eqref{Ebar} in terms of $\psi_\ell$ in $\Pbar$ for $\ell<m$, $\psi_\ell$ in $P$ for $\ell\leq m-2$, and $\ew_\ell$ for $\ell\leq m-2$.  Solvability of this equation is subject to two conditions.

The first condition is obtained by setting $v=0$ in $\Pbar$, and is simply the problem \eqref{E} for $\psi_{m-2}$ in $P$.  More precisely, given $v\in\Honeper(\QQ)$, decompose $v=v_0+v_1$ according to the orthogonal decomposition of $\Honeper(\QQ)$,
\begin{equation*}
  \Honeper(\QQ) = H^1_0(P) + H^1_0(P)^\perp,
\end{equation*}
where functions in $H^1_0(P)$ are considered to reside in $\Honeper(\QQ)$ by means of extension by zero.  This means that $v_0=0$ in $\Pbar$ and $\Delta v_1=v_1$ in $P$.  Now, given that $\psi_{m-2}$ satisfies \eqref{E} in $P$, we arrive at a Neumann problem for $\psi_m$ in $\Pbar$ by replacing $v$ with $v_1$.

The second solvability condition is obtained by setting $v=v_1=1$ in $\QQ$,
\begin{multline}\label{D}
  (-i)^{m-2}\ew_{m-2}\left( \int_\QQ\psi_0 \,+\, \ew_0\!\int_P\psi_* \right)
  + \int_\QQ \sum_{\ell=1}^{m-3} (-i)^\ell\ew_\ell\,\psi_{m-2-\ell}
   \;\;+ \\
  +\, \ew_0\int_P\tilde\psi_{m-2} \,-\, \ek^2\int_{\Pbar}\left( \kk\cdot\nabla\psi_{m-1}+\psi_{m-2} \right)
  + \int_P\left( \kk\cdot\nabla\psi_{m-3} + \psi_{m-4} \right) \,=\, 0.
\end{multline}
Given that this condition is satisfied, one can solve for $\psi_m$ in $\Pbar$ up to an additive constant, and we will always take the unique solution that satisfies the zero-average condition
\begin{equation}\label{zeroaverage}
  \int_\Pbar \psi_m = 0.
\end{equation}

\smallskip
\noindent
{\bfseries\slshape Solving for $\ew_m$.}\;
Equation \eqref{D} is to be viewed as the relation that determines the term $\ew_{m-2}$ of the power series expansion of the dispersion relation.  Indeed, the quantity multiplying $\ew_{m-2}$ depends only on the inclusion, the values of $\ek$ and $\kk$, and the choice of $\ew_0$, and it is always nonzero:
\begin{equation*}
  \int_\QQ\psi_0 \,+\, \ew_0\!\int_P\psi_* \,=\, |\Pbar| + \int_P(\psi_0+\ew_0\psi_*) \,
  =\, |\Pbar| + \int_P\sum_{n=1}^\infty \frac{\mu_n^2\av{\phi_n}{P}}{(\mu_n-\ew_0)^2}\,\phi_n
  = |\Pbar| + \sum_{n=1}^\infty \frac{\mu_n^2\av{\phi_n}{P}^2}{(\mu_n-\ew_0)^2} \,>\, 0.
\end{equation*}
Later on, we will prove that $\ew_m=0$ when $m$ is odd.

\smallskip
In order to solve inductively for all of the fields $\psi_m$ and coefficients $\ew_m$, we will proceed from the established ``base case" consisting of $\psi_0$ in $\QQ$ and $\ew_0$, both of which have been specified above, as well as $\psi_1$ and $\psi_2$ in $\Pbar$.
In $\Pbar$, $\psi_1$ is determined by equation \eqref{Ebar} for $m=1$,
\begin{equation}\label{psi1weak}
  \renewcommand{\arraystretch}{1.3}
\left\{
  \begin{array}{l}
    \displaystyle
    \ek^2\int_\Pbar (\nabla\psi_1+\kk)\cdot\nabla\bar v \,=\, 0  \qquad \forall v\in\Honeper(\Pbar), \\
    \vspace{-2ex} \\
    \displaystyle \int_\Pbar \psi_1 \,=\, 0,
  \end{array}
\right.
\end{equation}
which has a unique zero-mean solution.  Likewise, $\psi_2$ is the unique zero-mean solution to
\begin{equation}\label{psi2weak}
  \renewcommand{\arraystretch}{1.3}
\left\{
  \begin{array}{l}
    \displaystyle
    \ek^2\!\int_\Pbar (\nabla\psi_2+\kk)\cdot\nabla\bar v \,-\, \ek^2\!\int_\Pbar (\kk\cdot\nabla\psi_1+1)\bar v
    \,+\, \ew_0\!\int_\QQ\psi_0\bar v \,- \int_P \nabla\psi_0\cdot\nabla\bar v
    \,=\, 0  \qquad \forall v\in\Honeper(\Pbar), \\
    \vspace{-2ex} \\
    \displaystyle \int_\Pbar \psi_2 \,=\, 0,
  \end{array}
\right.
\end{equation}
which is solvable because $\psi_0$ satisfies \eqref{psi0strong} in $P$ and because of the solvability condition \eqref{solvability2}, which the pair $(\ek,\ew_0)$ is assumed to satisfy.

\begin{theorem}
If $\ew_0\not=\mu'_\ell$ for all $\ell$, there are unique functions $\psi_m\in\Honeper(\QQ)$ and numbers $\ew_m\in\CC$ such that $\psi_m=0$ and $\ew_m=0$ for $m<0$ and such that
\begin{enumerate}
  \item the fields $\psi_0$, $\psi_1$, and $\psi_2$ in $\Pbar$ and $\psi_0$ in $P$ and the number $\ew_0$ coincide with those already defined;
  \item equations \eqref{Ebar} and \eqref{D} are satisfied for all $m$;
  \item $\displaystyle\int_\Pbar\psi_m=0$ for all $m$;
  \item $\psi_m=\tilde\psi_m + (-i)^m\ew_m\psi_*$ in $P$, where $\tilde\psi_m$ satisfies \eqref{Etilde}.
\end{enumerate}

\end{theorem}

\smallskip
\begin{proof}
The proof is by induction on the following statement that depends on $n\geq2$:
 
 \smallskip
 \noindent
{\itshape
{Statement $S_n$:}\; For $n\geq2$, there exist
\begin{itemize}
  \item $\psi_m\in\Honeper(\QQ)$ for $m\leq n-2$ with $\int_\Pbar\psi_m=0$ if $m\geq1$,
  \item $\psi_m\in\Honeper(\Pbar)$ for $n-2<m\leq n$, with $\int_\Pbar\psi_m=0$,
  \item $\ew_m\in\CC$ for $m\leq n-2$,
\end{itemize}
that vanish for $m<0$ and coincide with $\psi_0$, $\psi_1$, and $\psi_2$ in $\Pbar$, $\psi_0$ in $P$, and $\ew_0$ as already defined, and such that \eqref{Ebar} is satisfied for $m\leq n$ and $\psi_m=\tilde\psi_m + (-i)^m\ew_m\psi_*$ in $P$ with $\tilde\psi_m$ satisfying \eqref{Etilde}.
} 
\smallskip

Statement $S_2$ holds by definition of $\psi_0$, $\psi_1$, and $\psi_2$ in $\Pbar$, $\psi_0$ in $P$, and $\ew_0$.  Given the validity of $S_m$ for $m\leq n$, we will prove $S_{n+1}$.

Notice that \eqref{Etilde}, with $m=n-1$ is an equation for $\tilde\psi_{n-1}$ with data that depend only on the functions $\psi_n$ and numbers $\ew_m$ already defined by the induction hypothesis.  Since $\ew_0$ is not a Dirichlet eigenvalue of $-\Delta$ in $P$, this problem has a unique solution $\tilde\psi_{n-1}$.  Let $\ew_{n-1}$ be defined by equation \eqref{D} with $m=n+1$.  Besides $\ew_m$, that expression depends only on the numbers $\ew_m$ and functions $\psi_m$ presumed in the induction hypothesis.  By definition of $\psi_*$, the function defined by $\psi_{n-1} := \tilde\psi_{n-1}+(-i)^{n-1}\ew_{n-1}\psi_*$ satisfies \eqref{E} in $P$ with $m=n-1$.  Moreover, since $\psi_*|_\dP=0$ and because of the boundary condition in \eqref{Etilde}, we have $\psi_{n-1}|_{\dP^-} = \psi_{n-1}|_{\dP^+}$, so that the functions $\psi_{n-1}$ in $\Pbar$ and $P$ together define a function $\psi_{n-1}$ that resides in $\Honeper(\QQ)$.

Now consider \eqref{Ebar} for $m=n+1$.  The two conditions of solvability are satisfied.  First, for all $v\in\Honeper(\QQ)$ with $v=0$ in $\Pbar$, this equation is satisfied by the definition of $\psi_{n-1}$.  Second, by setting $v=1$, one obtains the equation \eqref{D}, which holds because of the definition of $\ew_{n-1}$.  Thus \eqref{Ebar} for $m=n+1$ admits a unique solution $\psi_{n+1}\in\Honeper(\Pbar)$ subject to $\int_\Pbar \psi_{n+1} = 0$.

To complete the proof, we must show that \eqref{D} holds for all $m$.  This follows from setting $v=1$ in \eqref{Ebar}.
\end{proof}

\medskip
\noindent
{\bfseries Remark on interior eigenvalues}\label{subsec:eigenvalues}
\smallskip

\noindent
In the case that $\ew_0=\mu'_\ell$ for some $\ell$, we have seen that the higher-order cell problems in the foregoing analysis are not solvable if it is assumed that $u_0\not=0$ in the matrix.  Instead, one should set $u_0=0$ in $\Pbar$, which is in accordance with the fact that the eigenfunctions for the two-scale limit vanish in the matrix and are a multiple of the eigenfunction $\phi$ for $\mu'_\ell$ in the inclusion \cite{Zhikov2004}.  Now the leading-order term of the expansion of $\ew^\eta$ is fixed at $\mu'_\ell$, and one must solve for $\ew_m$ ($m>0$) and for all $\psi_m$.  Because $\ew_0$ is a Dirichlet eigenvalue in $P$, the solvability of $\psi_m$ is subject to a Fredholm condition both in $\Pbar$ (where one sets $v=1$) as well as in $P$ (where one sets $v=\phi$).  The solution is given up to an additive constant $\beta_m$ in $\Pbar$ and up to the addition of $\gamma_m\phi$ in $P$.  The numbers $\beta_m$, $\gamma_m$, and $\ew_m$ are determined by subsequent solvability conditions.  The procedure for solving for these constants and the fields $\psi_m$ inductively turns out to be quite complex, though not insurmountable, and we believe that an exponential bound can be obtained using the technique of generating functions that we illustrate in the next section.  For this communication, we are satisfied to concentrate on the case $\ew_0\not=\mu'_\ell$.

\section{Convergence of the power series}\label{sec:convergence}

Convergence of the formal power series for $u$ in $H^1(\QQ)$ and $\ew$ in $\CC$ is equivalent to an exponential bound on the norms $H^1$-norms of the functions $\psi_m$ and the numbers $\ew_m$.  An exponential bound is proved in Theorem~\ref{thm:convergence} below.
The iterative scheme for solving for $\psi_m$ and $\ew_m$ provides recursive upper bounds on their norms.  These bounds are based on the standard $H^1$ estimates obtained in the appendix.

\subsection{Bounds on the fields}

From \eqref{Ebar} and the bound \eqref{psiPbarbound} in Problem 2 of the appendix, with $K=\Omega^2_\Pbar\max\{1,A\}$, we obtain
\begin{multline}\label{est1}
  \ek^2 \|\psi_m\|_{H^1(\Pbar)}\leq K\bigg[\ek^2\left(2\|\psi_{m-1}\|_{H^1(\Pbar)} + \|\psi_{m-2}\|_{H^1(\Pbar)}\right) + \|\psi_{m-2}\|_{H^1(P)} + \\ 
  + 2\|\psi_{m-3}\|_{H^1(P)} + \|\psi_{m-4}\|_{H^1(P)} 
   + \sum_{\ell=0}^{m-2} |\ew_\ell| \left( \|\psi_{m-2-\ell}\|_{H^1(\Pbar)}+\|\psi_{m-2-\ell}\|_{H^1(P)} \right)  \bigg].
\end{multline}
From equation \eqref{Etilde} for $\tilde\psi$ and the bound \eqref{psiPbound}, we obtain
\begin{equation}\label{est2}
  \|\tilde\psi_{m}\|_{H^1(P)} \leq K_\ek\bigg[
   \|\psi_{m}\|_{H^1(\Pbar)} + 2\|\psi_{m-1}\|_{H^1(P)} + \|\psi_{m-2}\|_{H^1(P)}
   + \sum_{\ell=1}^{m-1}|\ew_\ell|\|\psi_{m-\ell}\|_{H^1(P)} \bigg],
\end{equation}
in which $K_\ek\leq C_1\ek^2 + C_2$ if $\ew_0$ is bounded away from the eigenvalues $\mu'_\ell$.

In \eqref{D}, the constant multiplying $\ew_{m-2}$ is at least $|\Pbar|$ for all $\ek$, and we obtain, by adjusting $K$ if necessary,
\begin{multline}\label{est3}
  |\ew_m| \leq K\bigg[  \ew_0\|\tilde\psi_{m}\|_{H^1(P)} +
  \ek^2\left( \|\psi_{m+1}\|_{H^1(\Pbar)}+\|\psi_{m}\|_{H^1(\Pbar)} \right) + \\
  + \|\psi_{m-1}\|_{H^1(P)}  + \|\psi_{m-2}\|_{H^1(P)} + \sum_{\ell=1}^{m-1} |\ew_\ell|\left( \|\psi_{m-\ell}\|_{H^1(P)}+\|\psi_{m-\ell}\|_{H^1(\Pbar)} \right)
  \bigg].
\end{multline}
From the decomposition $\psi_m = \tilde\psi_m + (-i)^m\ew_m\psi_*$ and the bound \eqref{psistarbound} on $\psi_*$, we obtain
\begin{equation}\label{est4}
  \|\psi_{m}\|_{H^1(P)} \leq \|\tilde\psi_{m}\|_{H^1(P)} + B_\ek |\ew_m|,
\end{equation}
in which $B_\ek=(B_1\ek^2+B_2)^2$.

\subsection{Proof of convergence}\label{subsec:convergence}

The foregoing analysis is most transparent using the quasistaticity $\eta$ as the expansion parameter, as we have done.  For the investigation of structures of fixed cell size with varying wavenumber $k$, it is more suitable to expand the fields and dispersion relation in powers of the cell size.  To this end, set
\begin{equation*}
  \rho = \frac{d}{\sqrt{\epsr}} = \frac{\eta}{\ek}.
\end{equation*}
Then we have the expansions
\begin{eqnarray*}
  u&=& u_0 + \rho\ek u_1 + \rho^2\ek^2 u_2 + \dots, \\
  \ew &=& \ew_0 + \rho\ek\ew_1 + \rho^2\ek^2\ew_2 + \dots.
\end{eqnarray*}
Recall the relation $u_m = i^m\bar u_0\psi_m$.

\begin{theorem}\label{thm:convergence}  
Let a branch $\ew_0=\ew_0(\kk,\ek)\in[\mu^*_m,\mu_m)$ of the homogenized dispersion relation be given.  For each $\ek\geq0$ such that $\ew_0\not=\mu'_\ell\,\forall\ell$,  there exist real positive numbers $C$ and $J$, independent of $\kk$, such that
\begin{equation}\label{bounds}
  \|\ek^m\psi_m\|_{H^1(\QQ)} < CJ^m \qquad\text{and}\qquad |\ek^m\ew_m| < CJ^m.
\end{equation}
For each $\ek_0>0$ and for each $\varepsilon>0$, $C$ and $J$ can be chosen such that \eqref{bounds} holds for all $\kk$ and for all $\ek$ such that $0\leq\ek\leq\ek_0$ and
\begin{equation}\label{nonresonance}
  \ew_0 \in [\mu^*_m,\mu_m)\setminus\bigcup\limits_{\ell}(\mu'_\ell-\varepsilon,\mu'_\ell+\varepsilon).
\end{equation}
If the band $[\mu_m^*,\mu_m)$ contains none of the $\mu'_\ell$, then $C$ and $J$ are independent of $\varepsilon$.

\end{theorem}

\begin{proof}
Let us simplify our notation for the frequency and the norms of the fields,
\begin{eqnarray*}
  && \bar p_m = \|\ek^m\psi_m\|_{H^1(\Pbar)},\\
  && p_m = \|\ek^m\psi_m\|_{H^1(P)},\\
  && \tilde p_m = \|\ek^m\tilde\psi_m\|_{H^1(P)}, \\
  && s_m = |\ek^m\ew_m|.
\end{eqnarray*}
The bounds (\ref{est1},\ref{est2},\ref{est3},\ref{est4}) yield the following bounds for $m\geq1$:
\begin{eqnarray}\label{inequalities}
  \renewcommand{\arraystretch}{1.5}
\left.
  \begin{array}{l}
    \bar p_{m+1} \leq K\Big(2\ek\bar p_m+\ek^2\bar p_{m-1} + p_{m-1}+2\ek p_{m-2}+\ek^2p_{m-3}
        + \textstyle\sum\limits_{\ell=0}^{m-1}s_\ell(\bar p_{m-1-\ell} + p_{m-1-\ell})\Big),\\
    \tilde p_m \leq K_\ek\Big( \bar p_m + 2\ek p_{m-1}+\ek^2p_{m-2} + \textstyle\sum\limits_{\ell=1}^{m-1}s_\ell\, p_{m-\ell} \Big),\\
    s_m \leq K\Big( \ew_0\,\tilde p_m + \ek\bar p_{m+1} + \ek^2\bar p_m + \ek p_{m-1} + \ek^2p_{m-2}
        +\textstyle\sum\limits_{\ell=1}^{m-1}s_\ell(\bar p_{m-\ell}+p_{m-\ell}) \Big),\\
     p_m \leq \tilde p_m + B_\ek s_m.
  \end{array}
\right.
\end{eqnarray}
To prove that the numbers $\bar p_m$, $\tilde p_m$, $s_m$, and $p_m$ are exponentially bounded, it suffices to prove that the sequences of positive numbers, defined by the system of recursion relations below, obtained by replacing inequality with equality in the relations above, are exponentially bounded.  Indeed, one observes that
$\bar p_m\leq\hat a_m$, $\tilde p_m\leq\hat b_m$, $s_m\leq\hat c_m$, and $p_m\leq\hat d_m$.  Recall that we have defined $\tilde\psi_0=0$, which gives $\tilde p_0=0$.
\begin{equation}\label{recursion1}
  \renewcommand{\arraystretch}{1.5}
\left.
  \begin{array}{ll}
  \hat a_m=\hat b_m=\hat c_m=\hat d_m=0 & (m<0),\\
   \hat a_0=\bar p_0,\; \hat a_1=\bar p_1,\; \hat b_0=\tilde p_0=0,\; \hat c_0=|\ew_0|,\; \hat d_0=p_0,& \\
    \hat a_{m+1} = K\Big(2\ek\hat a_m+\ek^2\hat a_{m-1} + \hat d_{m-1}+2\ek\hat d_{m-2}+\ek^2\hat d_{m-3}
        + \textstyle\sum\limits_{\ell=0}^{m-1}\hat c_\ell\big(\hat a_{m-1-\ell} + \hat d_{m-1-\ell}\big) \Big) & (m\geq1),\\
    \hat b_m = K_\ek\Big( \hat a_m + 2\ek\hat d_{m-1}+\ek^2\hat d_{m-2} + \textstyle\sum\limits_{\ell=1}^{m-1}\hat c_\ell \hat d_{m-\ell} \Big) & (m\geq1),\\
    \hat c_m = K\Big( \ew_0\,\hat b_m + \ek\hat a_{m+1} + \ek^2\hat a_m + \ek\hat d_{m-1} + \ek^2\hat d_{m-2}
        +\textstyle\sum\limits_{\ell=1}^{m-1}\hat c_\ell\big(\hat a_{m-\ell}+\hat d_{m-\ell} \big) \Big)& (m\geq1),\\
     \hat d_m = \hat b_m + B_\ek\hat c_m & (m\geq1).
  \end{array}
\right.
\end{equation}

In addition, it is proved in the appendix that $B_\ek = (B_1\ek^2+B_2)^2$ (equation \ref{psistarbound}) and that, under the condition \eqref{nonresonance}, $K_\ek\leq C_1\ek^2+C_2$ (equation \ref{Ktaubound}).
  This means that, if $\ek_1<\ek_2$, the values of $\{\bar p_m,\tilde p_m,s_m,p_m\}$ for $\ek=\ek_1$ are bounded above by the values of $\{\hat a_m,\hat b_m,\hat c_m,\hat d_m\}$ for $\ek=\ek_2$ (observe that the initial conditions are independent of $\ek$).  To prove the theorem, it therefore suffices to prove the sequences $\{\hat a_m,\hat b_m,\hat c_m,\hat d_m\}$ are exponentially bounded for each $\ek>0$.  The bound can be taken to be independent of $\kk$ by choosing $\hat a_1\geq\bar p_1$ for all $\kk$.

It is convenient to transform the recursion relations above by the change of variables
\begin{equation*}
  a_m = \hat a_m, \;
  b_m = \hat b_{m-1}, \;
  c_m = \hat c_{m-1}, \;
  d_m = \hat d_{m-1}.
\end{equation*}
The system for $\{\hat a_m,\hat b_m,\hat c_m,\hat d_m\}$ is equivalent to the following system for $\{a_m,b_m,c_m,d_m\}$:
\begin{equation}\label{recursion2}
  \renewcommand{\arraystretch}{1.5}
\left.
  \begin{array}{ll}
   a_m= b_{m+1}= c_{m+1}= d_{m+1}=0 & (m<0),\\
    a_0=\bar p_0,\;  a_1=\bar p_1,\;  b_1=\tilde p_0=0,\;  c_1=|\ew_0|,\;  d_1=p_0,& \\
    a_{m} = K\Big( 2\ek a_{m-1}+ \ek^2a_{m-2} +  d_{m-1}+ 2\ek d_{m-2}+ \ek^2d_{m-3}
        + \textstyle\sum\limits_{\ell=1}^{m-1} c_\ell\big( a_{m-1-\ell} +  d_{m-\ell}\big)\Big) & (m\geq2),\\
    b_m = K_\ek\Big( a_{m-1} + 2\ek d_{m-1}+\ek^2d_{m-2} + \textstyle\sum\limits_{\ell=2}^{m-1} c_\ell  d_{m+1-\ell} \big)\Big) & (m\geq2),\\
    c_m = K\Big( \ew_0\,b_m + \ek a_{m} + \ek^2a_{m-1} + \ek d_{m-1} + \ek^2d_{m-2}
        +\textstyle\sum\limits_{\ell=2}^{m-1} c_\ell\big( a_{m-\ell}+ d_{m+1-\ell} \big)\Big) & (m\geq2),\\
    d_m = b_m + B_\ek c_m & (m\geq2).
  \end{array}
\right.
\end{equation}
We shall prove that $\{a_m,b_m,c_m,d_m\}$ are exponentially bounded by proving that their generating functions have a nonzero radius of convergence.  Set
\begin{eqnarray*}
  && \textstyle\sum\limits_{n=0}^\infty a_nz^n = z\alpha(z) + a_0, \\
  && \textstyle\sum\limits_{n=0}^\infty b_nz^n = z\beta(z), \\
  && \textstyle\sum\limits_{n=0}^\infty c_nz^n = z\gamma(z), \\
  && \textstyle\sum\limits_{n=0}^\infty d_nz^n = z\delta(z),
\end{eqnarray*}
The system \eqref{recursion2} is equivalent to the following system among the formal power series $\alpha(z)$, $\beta(z)$, $\gamma(z)$, and $\delta(z)$
\begin{equation}\label{generating}
  \renewcommand{\arraystretch}{1.5}
\left.
  \begin{array}{l}
    z(\alpha(z)-a_1) = K\Big( 2\ek z^2\alpha(z) + \ek^2z^2(z\alpha(z)+a_0) + (z+2\ek z^2+\ek^2z^3)z\delta(z) \,+ \hspace{3em}\\
        \hfill +\, z(z\gamma(z))(z\alpha(z)+a_0) + (z\gamma(z))(z\delta(z))\Big),\\
    z(\beta(z)-b_1) = K_\ek\Big( z^2\alpha(z) + (2\ek z+\ek^2 z^2)z\delta(z) + z^{-1}(z(\gamma(z)-c_1))(z(\delta(z)-d_1)) \Big),\\
    z(\gamma(z)-c_1) = K\Big( \ek z(\alpha(z)-a_1) + \ew_0\,z(\beta(z)-a_1) +\ek^2z^2\alpha(z) + (\ek z+\ek^2z^2)z\delta(z) \,+ \hspace{3em}\\
        \hfill +\, (z(\gamma(z)-c_1))z\alpha(z) + z^{-1}(z(\gamma(z)-c_1))(z(\delta(z)-d_1)) \Big),\\
    z(\delta(z)-d_1) = z(\beta(z)-b_1) + B_\ek z(\gamma(z)-c_1).
  \end{array}
\right.
\end{equation}
Define the following functions of five variables:
\begin{equation*}
  \renewcommand{\arraystretch}{1.5}
\left.
  \begin{array}{l}
    A(\alpha,\beta,\gamma,\delta,z) = -(\alpha-a_1) + K\Big( 2\ek z\alpha + \ek^2z(z\alpha+a_0) + (z+2\ek z^2+\ek^2z^3)\delta \,+\\\hfill+\, z\gamma(z\alpha+a_0) + z\gamma\delta\Big),\\
    B(\alpha,\beta,\gamma,\delta,z) = -(\beta-b_1) + K_\ek\Big( z\alpha + (2\ek z+\ek^2 z^2)\delta + (\gamma-c_1)(\delta-d_1) \Big),\\
    C(\alpha,\beta,\gamma,\delta,z) = -(\gamma-c_1) + K\Big( \ek(\alpha-a_1) + \ew_0(\beta-a_1) +\ek^2z\alpha + (\ek z+\ek^2z^2)\delta \,+\hspace{2em}\\\hfill+\, z(\gamma-c_1)\alpha + (\gamma-c_1)(\delta-d_1) \Big),\\
    D(\alpha,\beta,\gamma,\delta,z) = -(\delta-d_1) + (\beta-b_1) + B_\ek (\gamma-c_1).
  \end{array}
\right.
\end{equation*}
One can check that $A$, $B$, $C$, and $D$ all vanish at $(\alpha,\beta,\gamma,\delta,z) = (a_1,b_1,c_1,d_1,0)$.  Moreover, the determinant of the Jacobian matrix with respect to the first four variables is nonzero at this point:
\begin{equation*}
  \left.\frac{\partial(A,B,C,D)}{\partial(\alpha,\beta,\gamma,\delta)}\right|_{z=0} = 
  \renewcommand{\arraystretch}{1.5}
\left[
  \begin{array}{cccc}
    -1 & 0 & 0 & 0 \\
    0 & -1 & K_\ek(\delta-d_1) & K_\ek(\gamma-c_1) \\
    K\ek & K\ew_0 & -1+K(\delta-d_1) & K(\gamma-c_1) \\
    0 & 1 & B_\ek & -1
  \end{array}
\right],
\end{equation*}
and thus we obtain
\begin{equation}\label{determinant}
  \det \frac{\partial(A,B,C,D)}{\partial(\alpha,\beta,\gamma,\delta)}(a_1,b_1,c_1,d_1,0) = 1 \not=0.
\end{equation}
By the implicit function theorem of analytic functions of several variables, we infer that there exist analytic functions represented by power series
\begin{eqnarray*}
  && \alpha(z) = a_1 + a_2z + a_3z^2 + \dots, \\
  && \beta(z) = b_1 + b_2z + b_3z^2 + \dots, \\
  && \gamma(z) = c_1 + c_2z + c_3z^2 + \dots, \\
  && \delta(z) = d_1 + d_2z + d_3z^2 + \dots,
\end{eqnarray*}
that are convergent in a neighborhood of $z=0$, say for $|z|<R>0$, and such that
$$(A,B,C,D)(\alpha(z),\beta(z),\gamma(z),\delta(z),z)=(0,0,0,0)$$ for $|z|<R$.  By definition of $A$, $B$, $C$, and $D$, the functions $\alpha$, $\beta$, $\gamma$, and $\delta$ satisfy \eqref{generating}, and their coefficients therefore satisfy \eqref{recursion2}.  By the convergence of the power series, the coefficients and therefore also $\hat a_m$, $\hat b_m$, $\hat c_m$, and $\hat d_m$ are exponentially bounded.
\end{proof}

\section{Power series solution of the cell problem}\label{sec:solution}

We are now ready to prove that the formal power series in fact converge to Bloch wave solutions and their dispersion relations for small but nonzero $\eta$, as long as $\ew_0(\kk,\ek)$ is not a Dirichlet eigenvalue of $\Delta$ in $P$.  In the process, we also show that the odd coefficients $\ew_m$ vanish.

\subsection{Solutions for nonzero $\eta$}

The exponential bound for the functions $\ek^m\psi_m$ and the numbers $\ek^m\ew_m$ and therefore also for the functions $\psi_m$ and numbers $\ew_m$, implies that the formal power series for the field $u$ and the square frequency $\ew$ converge for
$|\eta|<R$ for some $R>0$.  We show that the functions defined by these power series are indeed solutions of the cell problem, for each $\ek\geq0$.  For $|\eta|<R$, define the functions
\begin{eqnarray*}
  && u^\eta = u_0 + \eta u_1 + \eta^2 u_2 + \dots,\\
  && \ew^\eta = \ew_0 + \ew_1\eta + \ew_2\eta^2 + \dots,
\end{eqnarray*}
in $\QQ$, in which $u_m = i^m\bar u_0\psi_m$, and $\psi_m$ are the solutions to the cell problems already described.
For $v\in\Honeper(\QQ)$ and $\eta<R$, define
\begin{multline}
  a^\eta(v) := \ek^2\int_\Pbar \left( \nabla u^\eta\cdot\nabla\bar v + i\eta\kk\cdot(u^\eta\nabla\bar v-\bar v\nabla u^\eta) + \eta^2u^\eta\bar v \right)
  -\eta^2\ew^\eta\int_\Pbar u^\eta\bar v \;+\\
  + \eta^2 \int_P \left( \nabla u^\eta\cdot\nabla\bar v + i\eta\kk\cdot(u^\eta\nabla\bar v-\bar v\nabla u^\eta) + \eta^2u^\eta\bar v \right)
  -\eta^2\ew\int_P u^\eta\bar v.
\end{multline}
This function of $\eta$ has a convergent power series in $\eta$ that is obtained by inserting the series for $u^\eta$ and $\ew^\eta$ into the expression for $a^\eta(v)$ and expanding in powers of $\eta$.  The coefficients of this expansion are exactly the right-hand side of equation \eqref{Ebar}, multiplied by $i^m$, with the identification $u_m = i^m\bar u_0\psi_m$.  But this is equal to zero for all $m$ because the functions $\psi_m$ satisfy \eqref{Ebar}.  Thus all coefficients of the power-series expansion of $a^\eta(v)$ vanish.  We conclude that $a^\eta(v)=0$ for all $v\in\Honeper(\QQ)$ and $|\eta|<R$, which means that $(u^\eta,\ew^\eta)$ satisfies
\begin{equation*}
  a^\eta(v) = 0 \qquad \forall v\in\Honeper(\QQ),
\end{equation*}
which is the weak formulation \eqref{masterweak} of the PDE for $u$, for all $|\eta|<R$.

\begin{theorem}
In the solution $(u^\eta,\ew^\eta)$, all of the functions $\psi_m$, where $u_m = i^m\bar u_0\psi_m$, are real-valued, $\ew_m\in\RR$ for all $m$, and $\ew_m=0$ if $m$ is odd.
\end{theorem}

\begin{proof}
From the weak form of the cell problem \eqref{masterweak}, one shows that $\ew$ is necessarily real-valued by putting $v=u$, and therefore all coefficients $\ew_m$ are real.
One then proceeds by induction on the following set of statements for $n\geq2$:
\begin{enumerate}
  \item $\psi_m$ in $P$ is real-valued if $m\leq n-2$,
  \item $\psi_m$ in $\Pbar$ is real-valued if $m\leq n$,
  \item $\ew_m=0$ if $m\leq n-2$ and $m$ is odd.
\end{enumerate}
These statements hold for $n=2$, as one observes from equations (\ref{psi0},\ref{psi1strong},\ref{dispersion},\ref{psi2weak}) that $\psi_0$ in $\QQ$, $\psi_1$ and $\psi_2$ in $\Pbar$, and $\ew_0$ are all real-valued.  Let $n\geq2$ be arbitrary, and let us prove these statements for $n$ replaced by $n+1$.  By equation \eqref{Etilde}, with $m=n-1$, we see that $\tilde\psi_{n-1}$ is real-valued in $P$.  Equation \eqref{D} with $m=n+1$ shows that $(-i)^{n-1}\ew_{n-1}$ is real ($\psi_*$ is real-valued).  Since $\ew_{n-1}$ is real, $\ew_{n-1}=0$ if $n-1$ is odd.  Thus
$\psi_{n-1} = \tilde\psi_{n-1} + (-i)^{n-1}\ew_{n-1}\psi_*$ is real-valued.  Finally, equation \eqref{Ebar} for $m=n+1$ shows that $\psi_{n+1}$ is real-valued in $\Pbar$.
\end{proof}

\subsection{Convergence of as $\eta\to0$}

The solution of the cell problem yields a Bloch wave solution for our periodic structure,
\begin{equation*}
  U^\eta(x) = e^{ik\kk\cdot x} u(kx/\eta)
  = e^{ik\kk\cdot x} \Big( u_0(kx/\eta) + \eta\sum_{m=0}^\infty \eta^m u_{m+1}(kx/\eta) \Big),
\end{equation*}
which tends, in the sense of strong two-scale convergence, to a separable function that is a plane wave in $x$:
\begin{equation*}
  U^\eta(x) \stsc U_0(x,y) = e^{ik\kk\cdot x} u_0(y).
\end{equation*}
In the two-scale theory \cite{Zhikov2000},\cite[\S3]{Zhikov2004}, one writes the solution as a function of $x$ {plus} a function of $(x,y)$ that is supported in the soft phase $P$,
\begin{equation*}
  U_0(x,y) = u_0(x) + v_0(x,y) = e^{ik\kk\cdot x} + e^{ik\kk\cdot x} \ew_0\sum_{n=0}^\infty\frac{\av{\phi_n}{P}}{\mu_n-\ew_0}\phi_n(y),
\end{equation*}
in which the sum over $n$ is extended by zero into $\Pbar$ and we keep in mind that $\ew_0$ depends on $k$ and $\kk$.
The function $u_0(x)$ is a Bloch solution in a homogenized medium.  This homogeneous medium is realized approximately by a real periodic structure with small cell width $d$, in which the fields are measured in the matrix only, and the role of the inclusions is to lend the bulk medium its rich band-gap structure.  This idea is discussed in several other works such as \cite{BouchitteFelbacq2004,BouchitteSchweizer2008,FortesLiptonShipman2010,Smyshlyaev2009,KohnShipman2008,PendryHoldenRobbins1999}, in which the extreme properties of periodically dispersed inclusions hosted by a neutral matrix material produce meta-materials with interesting bulk properties.

\section{Negative material coefficient}\label{sec:negative}

If we take the material coefficient inside the inclusion $P$ to be large and negative, while retaining a small positive index in the matrix,
\begin{equation*}
  \apbar = 1, \qquad \ap = -\frac{d^2}{\epsr} = -\frac{\eta^2}{\epsr k^2},
\end{equation*}
one can follow the minus-sign through all of the foregoing calculations, obtaining a convergent power series for the field $u$ and frequency $\omega$, with analogs of all of the theorems we have proved.  The situation is actually simpler because now the equation in $P$ is always solvable and the nonresonance condition \eqref{nonresonance} is no longer needed.  We show in the appendix (see equations \ref{Cneg}, \ref{Bneg}, \ref{Kneg}) that, in place of the constants $K$, $K_\ek$, and $B_\ek$ in section \ref{subsec:convergence}, we may use a single constant $K$ that is independent of $\kk$ and $\ek$.  We therefore obtain a radius of convergence that is valid for all wavevectors.

In place of equations \eqref{psi0strong} and \eqref{psi0} for the field inside $P$, we obtain the Helmholtz equation with a change in sign in front of the frequency $\ew_0$,
\begin{equation}\label{psi0strongN}
  \renewcommand{\arraystretch}{1.4}
\left.
  \begin{array}{rl}
    \Delta \psi_0 - \ew_0\psi_0 = 0 & \text{ in } P, \\
    \psi_0\at{P} = 1 & \text{ on } \partial P,
  \end{array}
\right.
\end{equation}
and the solution is
\begin{equation}\label{psi0N}
  \psi_0(\yy) = \sum_{n=1}^\infty \frac{\mu_n\,\av{\phi_n}{P}}{\mu_n+\ew_0}\,\phi_n(\yy),
  \quad \yy\in P.
\end{equation}
The expression \eqref{solvability2} for the solvability for $\psi_2$ in $\Pbar$ remains unaltered, and, using the explicit expression for $\psi_0$, yields
\begin{equation}\label{dispersionN}
  \ek^2\!\int_\Pbar(\kdotgrad\psi_1+1) = \ew_0\sum_{n=1}^\infty \frac{\mu_n\av{\phi_n}{P}^2}{\mu_n+\ew_0}.
\end{equation}
Since the field $\psi_1$ in $\Pbar$ is decoupled from the inclusion, it does not depend on $\apbar$, and therefore this dispersion relation differs from that for positive $\apbar$ only by the plus-sign in front of $\ew_0$ in the denominator.  It possesses therefore only one branch, which passes through the origin in $(\ek,\ew_0)$-space, or equivalently, in $(k,\omega)$-space, as illustrated in Fig.~\ref{fig:DispersionNeg}.

The weak form \eqref{Ebar} of the sequence of PDEs in the unit cell becomes, for negative index,
\begin{multline}\label{EbarN}
  \ek^2\!\int_{\Pbar} \nabla\psi_m\cdot\nabla\bar v
  \,+\, \ek^2\!\int_{\Pbar}\left[ \kk\psi_{m-1}\cdot\nabla\bar v -\left(\kk\cdot\nabla\psi_{m-1}+\psi_{m-2}\right)\bar v \right] + \int_\QQ \sum_{\ell=0}^{m-2} (-i)^\ell\ew_\ell\,\psi_{m-2-\ell}\,\bar v \,+ \\
  +\int_P\nabla\psi_{m-2}\cdot\nabla v
  + \int_P \left[ \kk\psi_{m-3}\cdot\nabla\bar v - \left(\kk\cdot\nabla\psi_{m-3}+\psi_{m-4}\right)\bar v\right] \,=\, 0 \qquad
  \forall\; v\in \Honeper(\QQ),
\end{multline}
and the equations for $\psi$, $\tilde\psi$, and $\psi_*$ in the decomposition
$\psi_m = \tilde\psi_m + (-i)^m\ew_m\psi_*$ in $P$ become
\begin{equation}\label{EN}
  \renewcommand{\arraystretch}{1.5}
\left\{
  \begin{array}{l}
    \displaystyle \int_P\left( \nabla\psi_m\cdot\nabla\bar v + \ew_0\,\psi_m\bar v \right)
    + \int_P\left[ \kk\psi_{m-1}\cdot\nabla\bar v -\left( \kk\cdot\nabla\psi_{m-1}+\psi_{m-2} \right)\bar v \right] \,+\hspace{2em}\\
    \displaystyle \hfill +\int_P \sum_{\ell=1}^{m} (-i)^\ell \ew_\ell\,\psi_{m-\ell}
    \,\bar v \,=\, 0
    \qquad \forall\; v\in H^1_0(P),\\
    \psi_m|_{\dP^-} = \psi_m|_{\dP^+},
  \end{array}
\right.
\end{equation}
\begin{equation}\label{EtildeN}
  \renewcommand{\arraystretch}{1.5}
\left\{
  \begin{array}{l}
    \displaystyle \int_P\left( \nabla\tilde\psi_m\cdot\nabla\bar v + \ew_0\,\tilde\psi_m\bar v \right)
    +\int_P\left[ \kk\psi_{m-1}\cdot\nabla\bar v -\left( \kk\cdot\nabla\psi_{m-1}+\psi_{m-2} \right)\bar v \right] \,+\hspace{3em}\\
    \displaystyle \hfill +\int_P \sum_{\ell=1}^{m-1} (-i)^\ell \ew_\ell\,\psi_{m-\ell}
    \,\bar v \,=\, 0
    \qquad \forall\; v\in H^1_0(P),\\
    \tilde\psi_m|_{\dP^-} = \psi_m|_{\dP^+},
  \end{array}
\right.
\end{equation}
and
\begin{equation}\label{EstarN}
  \renewcommand{\arraystretch}{1.5}
\left\{
  \begin{array}{l}
  \displaystyle \int_P\left( \nabla\psi_*\cdot\nabla\bar v + \ew_0\,\psi_*\bar v \right)
  + \int_P\psi_0\bar v \,=\, 0 \qquad \forall\; v\in H^1_0(P),\\
   \psi_*|_{\dP} = 0.
  \end{array}
\right.
\end{equation}
The strong form of the latter equation is
\begin{equation*}
\renewcommand{\arraystretch}{1.3}
\left.
  \begin{array}{ll}
      (-\Delta + \ew_0)\psi_* = -\psi_0 & \text{ in } \; P, \\
      \psi_* = 0 & \text{ on } \; \dP.
  \end{array}
\right.
\end{equation*}
and the solution in terms of the Dirichlet eigenfunctions of $P$ is
\begin{equation}\label{psistarN}
  \psi_* = -\sum_{n=1}^\infty \frac{\mu_n\,\av{\phi_n}{P}}{(\mu_n+\ew_0)^2}\,\phi_n\,
  \qquad \text{in $P$.}
\end{equation}

The solvability condition for $\psi_m$ in $\Pbar$, which determines the value of $\ew_{m-2}$, is
\begin{multline}\label{DE}
  (-i)^{m-2}\ew_{m-2}\left( \int_\QQ\psi_0 \,+\, \ew_0\!\int_P\psi_* \right)
  + \int_\QQ \sum_{\ell=1}^{m-3} (-i)^\ell\ew_\ell\,\psi_{m-2-\ell}
   \;\;+ \\
  +\, \ew_0\int_P\tilde\psi_{m-2} \,-\, \ek^2\int_{\Pbar}\left( \kk\cdot\nabla\psi_{m-1}+\psi_{m-2} \right)
  - \int_P\left( \kk\cdot\nabla\psi_{m-3} + \psi_{m-4} \right) \,=\, 0.
\end{multline}
Again, the quantity in multiplying $\ew_{m-2}$ is always nonzero,
\begin{equation*}
  \int_\QQ\psi_0 \,+\, \ew_0\!\int_P\psi_* \,=\, |\Pbar| + \sum_{n=1}^\infty \frac{\mu_n^2\av{\phi_n}{P}^2}{(\mu_n+\ew_0)^2} \,>\, 0,
\end{equation*}
so that one can always solve for $\ew_{m-2}$.

\section{Appendix}\label{sec:appendix} 

In this appendix, we derive the bounds on the fields in $P$ and $\Pbar$ that are employed in the recursive system of inequalities in section \ref{subsec:convergence}.

\subsection{Bound for the resolvent of the Dirichlet-$\Delta$ in $P$}

Let $\{\varphi_j,\,\nu_j\}_{j=1}^\infty$, $0<\nu_1\leq\nu_2\leq\nu_3\leq\dots$, be the Dirichlet spectral data for the region $P$, where the $\varphi_j$ form an orthonormal Hilbert-space basis for $L^2(P)$, that is,
\begin{equation*}
  \renewcommand{\arraystretch}{1.5}
\left\{
  \begin{array}{ll}
    \displaystyle
    -\Delta \varphi_j=\nu_j \varphi_j & \mbox{in } P, \\
    \varphi_j=0 & \mbox{on }\partial P,\\
    \|\varphi_j\|_{L^2(P)} = 1. &
  \end{array}
\right.
\end{equation*}
The set $\{\nu_j\}$ is the disjoint union of the sets $\{\mu_n\}$ and $\{\mu'_\ell\}$, where the former consists of those eigenvalues for which there is an eigenfunction with nonzero mean. 
Using integration by parts, one finds that $\nabla\phi_j$ is orthogonal to $\nabla\phi_{j'}$ in the mean square inner product if $j\not=j'$ and that
\begin{equation}\label{innerproduct}
\Vert\nabla\varphi_j\Vert^2_{L^2(P)}=\nu_j\Vert\varphi_j\Vert^2_{L^2(P)} = \nu_j.
\end{equation}
Now let $u$ be the solution of 
\begin{equation}\label{problemP}
  \renewcommand{\arraystretch}{1.5}
\left\{
  \begin{array}{ll}
    \displaystyle
    \Delta u + \nu u = G & \mbox{in } P, \\
    u=0 & \mbox{on }\partial P,
  \end{array}
\right.
\end{equation}
in which $\nu$ is not an eigenvalue and $G\in L^2(P)$.  We shall obtain a bound on the solution $u$ in $H^1(P)$ in terms of $G$.

If the eigenfunction expansion of $G$ is
$$
G=\sum_{j=0}^{\infty}b_j\varphi_j,
$$
then
$$
u=\sum_{j=0}^{\infty}\frac{b_j}{\nu-\nu_j}\varphi_j
\;\;\mbox{ and }\;\;
\nabla u=\sum_{j=0}^{\infty}\frac{b_j}{\nu-\nu_j}\nabla\varphi_j,
$$
so that using \eqref{innerproduct} we obtain
$$
\Vert u\Vert^2_{L^2(P)}=\sum_{j=0}^{\infty}\left|\frac{b_j}{\nu-\nu_j}\right|^2
\;\;\mbox{ and }\;\;
 \Vert \nabla u\Vert^2_{L^2(P)}=\sum_{j=0}^{\infty}\nu_j\left|\frac{b_j}{\nu-\nu_j}\right|^2,
$$
Hence,
\begin{eqnarray*}
\Vert u\Vert^2_{L^2(P)}+\Vert \nabla u\Vert^2_{L^2(P)}&=&\sum_{j=0}^{\infty}(1+\nu_j)\left|\frac{b_j}{\nu-\nu_j}\right|^2\\ 
&\leq&\max\limits_{j}\left\{\frac{1+\nu_j}{|\nu-\nu_j|^2}\right\}\sum_{j=0}^{\infty}|b_j|^2.\\
\end{eqnarray*}
Since the latter sum is equal to $\Vert G\Vert^2_{L^2(P)}$, we obtain
\begin{equation*}
  \|u\|_{H^1(P)} \leq C_\nu \|G\|_{L^2(P)},
\end{equation*}
in which
\begin{equation}\label{Cmu}
  C_\nu=\max\limits_{j}\frac{(1+\nu_j)^\half}{|\nu-\nu_j|}.
\end{equation}
%

Now let us bound the constants $C_{\pm\ew_0}$ in terms of $\ek$.

In the case of negative material coefficient, the pair $(\ek,\ew_0)$ lies on the homogenized dispersion relation \eqref{dispersionN}, and we must put $\nu=-\ew_0$ in \eqref{problemP}.  Putting $\nu=-\ew_0$ in \eqref{Cmu} yields a uniform bound on $C_{-\ew_0}$:
\begin{equation}\label{Cneg}
  C_{-\ew_0} \leq C_{0} \quad \text{for }\; \ew_0\geq0.  
\end{equation}

In the case of positive material coefficient, $(\ek,\ew_0)$ lies on the homogenized dispersion relation \eqref{dispersion},
\begin{equation}\label{dispersion2}
  \ek^2 = B\ew_0\sum_{n=1}^\infty \frac{\mu_n\av{\phi_n}{P}^2}{\mu_n-\ew_0},
\end{equation}
in which $B = [\int_\Pbar(\kdotgrad\psi_1+1)]^{-1}$.  Let us choose the branch of this relation associated with the eigenvalue $\mu_m$ for some fixed $m$.  This means that $\ew_0(\ek)$ lies in the band $[\mu^*_m,\mu_m)$ for all $\ek$, where $\mu^*_m$ is the root of the right-hand-side of \eqref{dispersion2} that lies between $\mu_{m-1}$ and $\mu_m$ for $m>1$ and $\mu^*_0=0$.  The right-hand side of \eqref{dispersion2} is positive for $\mu^*_m<\ew_0<\mu_m$ and tends to infinity as $\ew_0$ tends to $\mu_m$ from below (see Fig.~\ref{fig:DispersionPos}).  Let $K$ be such that
\begin{equation*}
  \left| B\ew_0\sum_{n\not=m} \frac{\mu_n\av{\phi_n}{P}^2}{\mu_n-\ew_0} \right| < K,
\end{equation*}
so that
\begin{equation*}
  0 <
  B\ew_0 \frac{\mu_m\av{\phi_m}{P}^2}{\mu_m-\ew_0} \leq
  \ek^2 + K.
\end{equation*}
This yields the bound
\begin{equation}\label{inverse}
  \frac{1}{\mu_m-\ew_0} \leq
  \left[ B\ew_0\mu_m\av{\phi_m}{P}^2\right]^{-1}(\ek^2+K)
  \leq \left[ B\mu^*_m\mu_m\av{\phi_m}{P}^2\right]^{-1}(\ek^2+K).
\end{equation}
From \eqref{Cmu},
\begin{equation*}
  C_{\ew_0} \leq \max\left\{
  \frac{(1+\mu_m)^\half}{\mu_m-\ew_0}, M
  \right\},
\end{equation*}
in which
\begin{equation*}
    \frac{(1+\nu_j)^\half}{|\nu_j-\ew_0|} \leq M \qquad \text{if } \nu_j\not=\mu_m.
\end{equation*}

Here is precisely where the issue of secular terms in the series solution arises.  If there is an eigenvalue $\mu'_\ell$ in $[\mu^*_m,\mu_m)$, then $M\to\infty$ as $\ew_0\to\mu'_\ell$, and $M=\infty$ when $\ew_0=\mu'_\ell$.  In this case, there is no uniform bound on $C_{\ew_0}$ for bounded values of $\ek$ and, moreover, for that $\ek$ for which $\ew_0=\mu'_\ell$, the problem \eqref{problemP} is not well posed.

In the case in which there are no values $\mu'_\ell$ in $[\mu^*_m,\mu_m)$, $M$ can be taken to be a fixed number for all $\ew_0\in[\mu^*_m,\mu_m)$ and thus for all $\ek$.  More generally, if small neighborhoods of those eigenvalues $\mu'_\ell$ that are in $[\mu^*_m,\mu_m)$ are excised from this band, $M$ can be taken to be finite.
Thus, using our estimate for $C_{\ew_0}$ together with \eqref{inverse}, we obtain the bound
\begin{equation}\label{Czeta}
   C_{\ew_0} \leq A_1\ek^2 + A_2, \qquad \forall\,\ek\text{ such that } \ew_0(\ek) \in [\mu^*_m,\mu_m)\setminus\bigcup\limits_{\ell}(\mu'_\ell-\varepsilon,\mu'_\ell+\varepsilon).
\end{equation}
in which $A_1$ and $A_2$ depend on $P$ and the branch of the dispersion relation, and, in case there is an eigenvalue $\mu'_\ell$ in $[\mu^*_m,\mu_m)$, also on $\varepsilon$, and $A_2\to\infty$ as $\varepsilon\to0$.

\subsection{Bounds for the fields in $P$ and $\Pbar$}\label{subsec:appendixfields}

\noindent
{\slshape {\bfseries The field $\psi_*$}.}
In the negative case, we obtain, from the explicit expression \eqref{psistarN} for $\psi_*$,
\begin{equation}\label{Bneg}
  \|\psi_*\|^2_{H^1(P)} = \sum_{n=1}^\infty 
  \frac{(1+\mu_n)\mu_n^2\,\av{\phi_n}{P}^2}{(\mu_n+\ew_0)^4} < B,
\end{equation}
which is bounded by the single number $B$ for all $\ew_0\geq0$.

In the positive case, the expression \eqref{psistar} of $\psi_*$ yields
\begin{equation*}
  \|\psi_*\|^2_{H^1(P)} = \sum_{n=1}^\infty 
  \frac{(1+\mu_n)\mu_n^2\,\av{\phi_n}{P}^2}{(\mu_n-\ew_0)^4}.
\end{equation*}
Since all of these summands except for $n=m$ are uniformly bounded over $n$ and $\ew_0\in[\mu^*_m,\mu_m)$, we obtain
\begin{equation*}
   \|\psi_*\|_{H^1(P)} \leq K + \frac{(1+\mu_m)^\half\mu_m\,\av{\phi_m}{P}}{(\mu_m-\ew_0)^2},
\end{equation*}
in which $K$ depends only on $P$ and the branch of the dispersion relation.  In view of \eqref{inverse}, we obtain, for some constants $B_1$ and $B_2$,
\begin{equation}\label{psistarbound}
  \|\psi_*\|_{H^1(P)} \leq (B_1\ek^2 + B_2)^2 =: B_\ek.
\end{equation}

\medskip
\noindent
{\slshape {\bfseries Problem 1:} Field extension from $\Pbar$ to $\QQ$ and solution in $P$.}\,
Suppose that $\psi$ is defined as an $H^1$ function in $\Pbar$, $G\in L^2(P)$, and $\ew$ is not a Dirichlet eigenvalue of $-\Delta$ in $P$.  Consider the extension of $\psi$ to $H^1(\QQ)$ such that
\begin{equation*}
  \renewcommand{\arraystretch}{1.5}
\left\{
  \begin{array}{ll}
    \displaystyle
    \int_P(\nabla\psi\cdot\nabla\bar v-\ew\psi\bar v) = \int_P G\bar v \qquad
    \forall\, v\in H^1_0(P), \\
    \psi|_{\dP^-} = \psi|_{\dP^+}.
  \end{array}
\right.
\end{equation*}
Decompose $\psi$ as $\psi=\psi_0+\psi_1$ such that
\begin{equation*}
  \renewcommand{\arraystretch}{1.5}
\left\{
  \begin{array}{ll}
    \displaystyle
    \int_P(\nabla\psi_1\cdot\nabla\bar v+\psi_1\bar v) = 0 \qquad
    \forall\, v\in H^1_0(P), \\
    \psi_1|_{\dP^-} = \psi|_{\dP^+}.
  \end{array}
\right.
\end{equation*}
and
\begin{equation*}
  \renewcommand{\arraystretch}{1.5}
\left\{
  \begin{array}{l}
    \displaystyle
    \int_P(\nabla\psi_0\cdot\nabla\bar v-\ew\psi_0\bar v)
    = \int_P (G+(1+\ew)\psi_1)\bar v \qquad
    \forall\, v\in H^1_0(P), \\
    \psi_0|_{\dP^-} = 0.
  \end{array}
\right.
\end{equation*}
We obtain
\begin{eqnarray*}
  && \|\psi_1\|_{H^1(P)} \leq \text{const}\cdot\|\psi\|_{H^\half(\dP)} \leq A\|\psi\|_{H^1(\Pbar)}, \\
  && \|\psi_0\|_{H^1(P)} \leq C_\ew(\|G\|_{L^2(P)} + |1+\ew|\|\psi_1\|_{H^1(P)}),
\end{eqnarray*}
in which, as we have discussed, $C_\ew$ is large when $\ew$ is close to a Dirichlet eigenvalue of $-\Delta$ in $P$.
From these two bounds, we obtain
\begin{equation*}
  \|\psi\|_{H^1(P)} \leq C_\ew\|G\|_{L^2(P)} + A(C_\ew|1+\ew|+1)\|\psi\|_{H^1(\Pbar)}.
\end{equation*}

Now, in the case of negative material coefficient, when $(\ek,\ew_0)$ lies on the homogenized dispersion relation, we must put $\ew=-\ew_0$.  Because of \eqref{Cneg} we obtain a bound
\begin{equation}\label{Kneg}
    \|\psi\|_{H^1(P)} \leq K (\|G\|_{L^2(P)} + \|\psi\|_{H^1(\Pbar)})
\end{equation}
that holds for all $\kk$ and $\ek$.

In the positive case, we must put $\ew=\ew_0$, where $(\ek,\ew_0)$ lies on a chosen branch of the dispersion relation, and we obtain
\begin{equation}\label{psiPbound}
  \|\psi\|_{H^1(P)} \leq K_\ek (\|G\|_{L^2(P)} + \|\psi\|_{H^1(\Pbar)}),
\end{equation}
in which
\begin{equation*}
  K_\ek = \max\{C_{\ew_0}, A(C_{\ew_0}|1+\mu_m|+1)\}.
\end{equation*}
If the branch contains none of the $\mu'_\ell$ or these numbers are excised, the bound \eqref{Czeta} yields
\begin{equation}\label{Ktaubound}
  K_\ek \leq C_1\ek^2 + C_2 \qquad \forall\,\ek\text{ such that } \ew_0(\ek) \in [\mu^*_m,\mu_m)\setminus\bigcup\limits_{\ell}(\mu'_\ell-\varepsilon,\mu'_\ell+\varepsilon).
\end{equation}
in which $C_1$ and $C_2$ again depend on $P$ and the branch, and, in case there is an eigenvalue $\mu'_\ell$ in $[\mu^*_m,\mu_m)$, also on $\varepsilon$, and $C_2\to\infty$ as $\varepsilon\to0$.

\medskip
\noindent
{\slshape {\bfseries Problem 2:} Solution in $\Pbar$.}\, 
Suppose that the function $\psi\in\Honeper(\Pbar)$ satisfies
\begin{equation*}
  \renewcommand{\arraystretch}{1.5}
\left\{
  \begin{array}{l}
    \displaystyle
    \int_\Pbar \nabla\psi\cdot\nabla\bar v = \int_\Pbar(F_1\cdot\nabla\bar v+G_1\bar v)
      + \int_P(F_2\cdot\nabla\bar v+G_2\bar v) \qquad \forall\, v\in\Honeper(\QQ), \\
    \vspace{-3ex}\\
    \displaystyle
    \int_\Pbar\psi = 0.
  \end{array}
\right.
\end{equation*}
The following solvability conditions are necessarily satisfied,
\begin{equation*}
  \renewcommand{\arraystretch}{1.5}
\left\{
  \begin{array}{l}
    \displaystyle
    \int_P(F_2\cdot\nabla\bar v + G_2\bar v) = 0 \quad\forall\, v\in H^1_0(P),\\
    \vspace{-3ex}\\
    \displaystyle
    \int_\Pbar G_1 + \int_P G_2 = 0.
  \end{array}
\right.
\end{equation*}
The zero-average condition implies the Poincar\'e bound
\begin{equation*}
  \|\psi\|_{H^1(\Pbar)} \leq \Omega_\Pbar\|\nabla\psi\|_{L^2(\Pbar)}.
\end{equation*}
Extend $\psi$ to $\Honeper(\QQ)$ such that $\Delta\psi=\psi$ in $P$ so that
\begin{equation*}
  \|\psi\|_{H^1(P)} \leq A\|\psi\|_{H^1(\Pbar)},
\end{equation*}
and set $v=\psi$ above.
\begin{multline*}
  \|\nabla\psi\|^2_{L^2(\Pbar)} \leq \big( \|F_1\|_{L^2(\Pbar)} + \|G_1\|_{L^2(\Pbar)} \big)\|\psi\|_{H^1(\Pbar)}
     + \left( \|F_2\|_{L^2(P)} + \|G_2\|_{L^2(P)} \right)\|\psi\|_{H^1(P)} \\
     \leq \left( \|F_1\|_{L^2(\Pbar)} + \|G_1\|_{L^2(\Pbar)} + A\left( \|F_2\|_{L^2(P)} + \|G_2\|_{L^2(P)} \right)\right)\|\psi\|_{H^1(\Pbar)}.
\end{multline*}
Using the Poincar\'e bound, we then obtain
\begin{equation}\label{psiPbarbound}
  \|\psi\|_{H^1(\Pbar)} \leq \Omega_{\Pbar}^2\left( \|F_1\|_{L^2(\Pbar)} + \|G_1\|_{L^2(\Pbar)} + A\left( \|F_2\|_{L^2(P)} + \|G_2\|_{L^2(P)} \right) \right).
\end{equation}

\bigskip
{\bfseries Acknowledgment.}\
{R. Lipton was supported by NSF grant DMS-0807265 and AFOSR grant FA9550-05-0008, and S. Shipman was supported by NSF grant DMS-0807325.  The Ph.D. work of S. Fortes was also supported by these grants.  This work was inspired by the IMA ``Hot Topics" Workshop on Negative Index Materials in October, 2006.}

\bibliography{FLS2}

\end{document}